\documentclass[11pt]{article}
\usepackage{fullpage,color,xspace,hyperref}
\usepackage{amsthm,amsmath,amssymb}
\usepackage{natbib}

\newtheorem{theorem}{Theorem}[section]
\newtheorem{lemma}[theorem]{Lemma}

\newtheorem{corollary}[theorem]{Corollary}

\newtheorem{definition}[theorem]{Definition}

\newtheorem{problem}[theorem]{Problem}
\newtheorem{remark}[theorem]{Remark}

\newcommand{\RP}{\ensuremath{\mathsf{RP}}}
\newcommand{\NP}{\ensuremath{\mathsf{NP}}}
\newcommand{\R}{\mathop{\mathbf{R}}}
\newcommand{\E}{\mathop{\mathbf{E}}}
\renewcommand{\Pr}{\mathop{\mathbf{Pr}}}
\newcommand{\VSTAT}{{\mbox{VSTAT}}}
\newcommand{\STAT}{{\mbox{STAT}}}
\newcommand{\SD}{{\mathrm{SD}}}
\newcommand{\SDA}{{\mathrm{SDA}}}
\newcommand{\HSQ}{{\mathrm{HSQ}}}
\newcommand{\SAMPLE}{{\mbox{1-STAT}}}

\newcommand{\D}{{\mathcal D}}
\newcommand{\F}{{\mathcal F}}
\newcommand{\Z}{{\mathcal Z}}

\newcommand{\eps}{\epsilon}
\newcommand{\B}{{\mathcal B}}
\newcommand{\C}{{\mathcal C}}
\newcommand{\ra}{\rangle}
\newcommand{\la}{\langle}
\newcommand{\A}{{\mathcal A}}

\newcommand{\sgn}{\mathsf{sign}}

\newcommand{\SQDIM}{\mbox{SQ-DIM}\xspace}
\newcommand{\eat}[1]{}
\newcommand{\norm}[1]{\left\|#1\right\|}
\newcommand{\abs}[1]{\left|#1\right|}
\newcommand{\ones}[1]{\norm{#1}_1}
\newcommand{\zeros}[1]{\norm{\overline{#1}}_1}
\newcommand{\innerprod}[2]{\left<{ #1},{#2}\right>}
\newcommand{\cald}{\mathcal D}
\newcommand{\apbc}{\textsc{APBC}}
\newcommand{\dpbc}{\textsc{DPBC}}
\newcommand{\zo}{\{0,1\}}
\newcommand{\cond}{\ |\ }
\newcommand{\Bin}{B}
\newcommand{\cS}{{\mathcal S}}
\newcommand{\sampleoracle}{$1$-bit sampling\xspace}

\newcommand{\equ}[1]{
\begin{equation}
#1
\end{equation}}

\newcommand{\alequn}[1]{\begin{align*} #1 \end{align*}}

\providecommand{\eg}{{\em e.g.}~}

\newcommand{\citeAN}{\citet}
\newcommand{\citeP}{\citep}
\newcommand{\citeY}{\citeyearpar}

\def\mnotes{1} % 0 means hide notes
\ifnum\mnotes=0 %set to 0 disables the side notes
\newcommand{\mnote}[1]{}
\else
        \newcounter{mynotes}
        \newcommand{\mnote}[1]{\addtocounter{mynotes}{1}{{\bf !}}
        \marginpar{\scriptsize  {\arabic{mynotes}.\ {\sf \textcolor{blue}{#1}}}}}
        \fi

\begin{document}

\begin{titlepage}
\title{Statistical Algorithms and a Lower Bound for Detecting Planted Cliques}

\eat{
\author{
Vitaly Feldman
\qquad Elena Grigorescu\thanks{Research supported in part by the NSF Grant \#1019343 to the Computing Research Association for the
CIFellows Project.}
\thanks{
Research supported in part by NSF awards
CCF-0915903 and CCF-1217793.
}
\qquad Lev Reyzin\thanks{Research supported by a Simons Postdoctoral Fellowship.}
\footnotemark[2]\\ \\
 Santosh S.\ Vempala\footnotemark[2]
\qquad Ying Xiao\footnotemark[2]\\
\\ \\
\texttt{vitaly@post.harvard.edu}\\
Almaden Research Center\\
IBM\\
San Jose, CA 95120\\ \\
\texttt{elena-g@purdue.edu}\\
Department of Computer Science\\
Purdue University\\
West Lafayette, IN 47907\\ \\
\texttt{lreyzin@math.uic.edu}\\
Department of Mathematics, Statistics, and Computer Science\\
University of Illinois at Chicago\\
Chicago, IL 60607\\ \\
\texttt{\{vempala,yxiao32\}@cc.gatech.edu}\\
School of Computer Science\\
Georgia Institute of Technology\\
Atlanta, GA 30332}
}

\author{
Vitaly Feldman\thanks{IBM Almaden Research Center.}
\qquad Elena Grigorescu\thanks{Department of Computer Science,
Purdue University. Supported by the National Science Foundation under Grant \#1019343 to the Computing Research Association for the
CIFellows Project.}
\thanks{
Research supported in part by NSF awards
CCF-0915903 and CCF-1217793.
}
\qquad Lev Reyzin\thanks{Department of Mathematics, Statistics, and Computer Science,
University of Illinois at Chicago. Research supported by a Simons Postdoctoral Fellowship.}
\footnotemark[3]\\ \\
 Santosh S.\ Vempala\thanks{School of Computer Science,
Georgia Institute of Technology. Research supported by NSF awards CCF-1217793 and EAGER-1415498.}~\footnotemark[3]
\qquad Ying Xiao\footnotemark[5]~\footnotemark[3]
}

\date{}

\maketitle
\thispagestyle{empty}

\begin{abstract}
  We introduce a framework for proving lower bounds on computational
  problems over distributions against algorithms that can be implemented using
  access to a {\em statistical query} oracle.   For such algorithms, access to the
  input distribution is limited to obtaining an estimate of the
  expectation of any given function on a sample drawn randomly from
  the input distribution, rather than directly accessing samples.
  Most natural algorithms of interest in theory and in practice,
  e.g.,\ moments-based methods, local search, standard iterative
  methods for convex optimization, MCMC and simulated annealing can be
  implemented in this framework. Our  framework is based on, and generalizes, the statistical query model
  in learning theory \citeP{Kearns:98}.

  Our main application is a nearly optimal lower bound on the complexity of {\em any} statistical query algorithm for detecting planted
  bipartite clique distributions (or planted dense subgraph
  distributions) when the planted clique has size $O(n^{1/2-\delta})$
  for any constant $\delta > 0$.  The assumed hardness of variants of these problems
has been used to prove hardness of several other problems and as a guarantee for security in cryptographic applications.
  Our lower bounds provide concrete
  evidence of hardness, thus supporting these assumptions.
\end{abstract}
\end{titlepage}

\tableofcontents
\newpage
\section{Introduction}

We study the complexity of problems where the input consists of independent samples from an unknown distribution.
Such problems are at the heart of machine learning and statistics (and their numerous applications) and also occur in many other contexts such as compressed sensing and cryptography.
While several methods have been developed to estimate the sample complexity of such problems (e.g.~VC dimension  \citeP{VapnikChervonenkis:71} and Rademacher complexity \citeP{BartlettMendelson:02}), proving lower bounds on the computational complexity of these problems has been much more challenging.
 The traditional approach to proving lower bounds  is via reductions and by finding distributions that can generate instances of some problem conjectured to be intractable (e.g., assuming $\NP \neq \RP$).

Here we present a different approach. We show that algorithms which access the unknown distribution only via a {\em statistical query (SQ) oracle} have high complexity, unconditionally. Most algorithmic approaches used in practice and in theory on a wide variety of problems can be
implemented using only access to such an oracle; these include Expectation Maximization (EM)~\citeP{DempsterLR77},
local search, MCMC optimization \citeP{TannerW87,GelfandSmith90},
simulated annealing~\citeP{KirkpatrickGV83,Cerny1985Thermodynamical},
first and second order methods for linear/convex optimization, ~\citeP{DunaganV08,BelloniFV09}, $k$-means, Principal Component Analysis (PCA), Independent Component Analysis
(ICA), Na\"{\i}ve Bayes, Neural Networks and many others (see \citeP{ChuKLYBNO:06} and \citeP{BlumDMN:05} for proofs and many other examples). In fact, we are aware of only one
algorithm that provably does not have a statistical query counterpart: Gaussian elimination for solving linear equations over a field (\eg$\mod 2$).

Informally, a statistical query oracle provides an estimate of the expected value of any given bounded real-valued function within some tolerance. Many popular algorithms rely only on the average value of various functions over random samples (commonly referred to as
{\em empirical averages}). Standard Chernoff-Hoeffding bounds imply that the average value of a bounded function on the independent samples will be highly concentrated around the
expectation on the unknown distribution (and, indeed in many cases the empirical average is used precisely to obtain an estimate of the expectation). As a result such algorithms can
often be equivalently analyzed in our oracle-based model.

Our approach also allows proving lower bounds against algorithms that rely on a {\em \sampleoracle} oracle, referred to as \sampleoracle algorithms. This oracle provides the value of any Boolean function on a fresh random sample from the
distribution. Many existing algorithms require only such limited access to random samples. Others can be implemented using such access to samples (possibly using a polynomially larger number of samples). For brevity, we refer to algorithms that rely on either of these types of oracles as {\em statistical algorithms}.

%are algorithms that only use samples to obtain values $\sum_i g(x_i)$, where $x_i$'s are the samples available to the algorithm
%and $g$ is any real-valued function. Formally, such a definition would not actually restrict the power of the algorithm and
For example, many problems over distributions are solved using convex programs. Such a problem is typically formulated as finding an approximation to $\min_{z \in K} \E_{x \sim D}[f(x,z)]$
for some convex set $K$ and functions $f(x,\cdot)$ that are convex in the second parameter $z$. A standard approach (both in theory and practice) to solve such a problem is to use a
gradient descent-based technique.  The gradient of the
objective function is % (by interchanging the derivative with the integral)
\[
\nabla_z \E_x[f(x,z)] = \E_x[\nabla_z f(x,z)]
\]
and is usually estimated using the average value of $\nabla_z f(x,z)$ on (some of) the given random samples. However, standard analysis of gradient descent-based algorithms implies that a sufficiently accurate estimate of each of the coordinates of $\E_x[\nabla_z f(x,z)]$ would also suffice. Hence, for an objective function of the form above, gradient descent
can be implemented using either of the above oracles (detailed analysis of such implementations can be found in a subsequent work \citeP{FeldmanGV:15}).

\eat{
 A more involved example is Linear Programming. One version is the feasibility problem: find a nonzero $w$ s.t. $x\cdot w \ge 0$ for all $x$ in some
set $A$. We can formulate this as
\[
\max_w \E_{x\sim D} [ \sgn(x \cdot w)]
\]
and the distribution $D$ could be uniform over the set $A$ if it is finite. This problem can be solved by a SQ algorithm \citeP{BlumFKV98, DunaganV08}. This is also the case for
semidefinite programs and for conic optimization \citeP{BelloniFV09}. }
The key motivation for our framework is the empirical observation that almost all algorithms that work on random
samples are either already statistical in our sense or have natural statistical counterparts. Thus,  lower bounds for statistical algorithms can be directly translated into lower bounds
against a large number of existing approaches. %We remark that
%and strongly indicate the need for new approaches even for explicit instances.
We present the formal oracle-based definitions of statistical algorithms in Section \ref{sec:defs}.

%These results rule out many natural approaches to solving these problems in theory and provide some practical guidance about when not to use popular and generic heuristics such as EM or simulated annealing. Our work also serves to highlight the question: what nonstatistical algorithms exist for search and optimization problems?

Our model is based on the {\em statistical query learning}
model  \citeP{Kearns:98} defined as a restriction of
Valiant's \citeY{Valiant84} PAC learning model. The primary goal of
the restriction was to simplify the design of noise-tolerant learning
algorithms. As was shown by Kearns and others in subsequent works,
almost all classes of functions that can be learned efficiently can
also be efficiently learned in the SQ model. A notable and
so far unique exception is the algorithm for learning parities,
based on Gaussian elimination. As was already shown by
\citeAN{Kearns:98}, parities require exponentially many queries to learn in the
SQ model. Further, \citeAN{BlumFJKMR94} proved that the number of SQs
required for weak learning (that is, for obtaining a non-negligible
advantage over the random guessing) of a class of functions $C$ over a
fixed distribution $D$ is characterized by a combinatorial parameter
of $C$ and $D$, referred to as SQ-DIM$(C,D)$, the SQ
dimension.

We consider SQ algorithms in the broader context of arbitrary computational problems over distributions. We also define an SQ oracle that strengthens the oracle introduced by \citeAN{Kearns:98}.
For any problem over distributions we define a parameter of the problem that lower bounds the complexity of solving the problem by any SQ algorithm in the same way
that SQ-DIM lower bounds the complexity of learning in the SQ model. Our techniques for proving lower bounds are also based on methods developed for
lower-bounding the complexity of SQ learning algorithms. However, as we will describe later, they depart from the known techniques in a number of significant ways that are necessary for our
more general setting and our applications.

The \sampleoracle oracle and its more general $k$-bit version was introduced by \citeAN{Ben-DavidD98}. They showed that it is equivalent (up to polynomial factors) to the SQ oracle. Using our stronger SQ oracle we sharpen this equivalence. This sharper relationship is crucial for  obtaining meaningful lower bounds against $1$-bit sampling algorithms in our applications.

We demonstrate our techniques by applying them to the problems of detecting planted bipartite cliques and planted bipartite dense subgraphs. We now define these problems precisely and give some background.

\smallskip \noindent{\bf Detecting Planted Cliques.}  In the planted clique
problem, we are given a graph $G$ whose edges are generated by
starting with a random graph $G_{n,1/2}$, then ``planting," i.e.,
adding edges to form a clique on $k$ randomly chosen vertices.
\citeAN{Jerrum92} and \citeAN{Kucera95} introduced the planted clique problem as a potentially
easier variant of the classical problem of finding the largest clique
in a random graph \citeP{Karp:79}. A random graph $G_{n,1/2}$ contains
a clique of size $2\log n$ with high probability, and a simple greedy
algorithm can find one of size $\log n$. Finding cliques of size
$(2-\eps)\log n$ is a hard problem for any $\eps > 0$. Planting a
larger clique should make it easier to find one. The problem of
finding the smallest $k$ for which the planted clique can be detected
in polynomial time has attracted significant attention. For $k\geq c
\sqrt{n\log n}$, simply picking vertices of large degrees suffices
\citeP{Kucera95}. Cliques of size $k=\Omega(\sqrt{n})$ can be found
using spectral methods ~\citeP{AlonKS98,McSherry01,Coja10}, via SDPs \citeP{FeigeK00}, combinatorial methods \citeP{FeigeRon:10,DekelGP:11}, nuclear norm minimization
\citeP{AmesV11} and belief propagation \citeP{DeshpandeMontanari13}.

While there is no known polynomial-time
algorithm that can detect cliques of size below the threshold of
$\Omega(\sqrt{n})$, there is a
quasipolynomial algorithm for any $k \ge 2\log n$: enumerate subsets of size $2\log n$; for each subset that forms a clique, take all common neighbors of the subset; one of these will
be the planted clique. This is also the fastest known algorithm for any $k = O(n^{1/2-\delta})$, where $\delta > 0$.

Some evidence of the hardness of the problem was shown by
\citeAN{Jerrum92} who proved that a specific approach using a Markov
chain cannot be efficient for $k=o(\sqrt{n})$. Additional evidence of
hardness is given in \citeP{feige2003probable}, where it is shown that
Lov\'{a}sz-Schrijver SDP relaxations, which include the SDP used in
\citeP{FeigeK00}, cannot be used to efficiently find cliques of size
$k=o(\sqrt{n})$. Most recently, lower bounds against a constant level of the more powerful Sum-of-Squares SDP hierarchy were shown by \citeAN{MekaPW15} and \citeAN{DeshpandeM15}.
  The problem has been used to generate cryptographic
primitives~\citeP{JuelsP00}, and as a hardness assumption in a large number of works (\eg \citeP{AlonAKMRX07,HazanK11,MV:09,BerthetR:13,Dughmi14}).
%demonstrate the hardness of finding approximate Nash equilibria of certain games

We focus on the bipartite planted clique problem, where a $(k \times k)$-biclique is planted in a random bipartite graph. A densest-subgraph version of the bipartite planted clique
problem has been used as a hard problem for cryptographic applications \citeP{ApplebaumBW10}.
%It is not hard to see that finding a planted $k$-biclique is at least as hard as finding the planted $k$-clique. At the same time
The bipartite version can be easily seen to be at least as hard as the original version. At the same time all known bounds and algorithms for the $k$-clique problem can be easily adapted
to the bipartite case (\eg \citeP{AmesV11}). Therefore it is natural to expect that new upper bounds on the planted $k$-clique problem would also yield new upper bounds for the
bipartite case.

The starting point of our investigation for this problem is the property of the planted $k$-biclique problem that it has an equivalent formulation as a problem over distributions
defined as follows.
\begin{problem}\label{def:clique}
Fix an integer $k$, $1 \le k \le n$, and a subset of $k$ indices $S\subseteq \{1,2,\ldots,n\}$.
The input distribution $D_S$ on vectors $x \in \{0,1\}^n$ is defined as follows:
with
probability $1-(k/n)$, $x$ is uniform over $\{0,1\}^n$; and with probability $k/n$, $x$ is such that its $k$ coordinates from $S$ are set to $1$,
and the remaining coordinates are uniform in $\{0,1\}$. For an integer $t$, the \textbf{distributional planted $k$-biclique} problem with $t$ samples is the problem of finding the
unknown subset $S$ using $t$ samples drawn randomly from $D_S$.
\end{problem}

One can view samples $x_1,\ldots,x_t$ as adjacency vectors of the vertices of a
bipartite graph as follows: the bipartite graph has $n$ vertices on the right (with $k$ marked as members of the clique) and $t$ vertices on the left. Each of the $t$ samples gives the
adjacency vector of the corresponding vertex on the left. It is not hard to see that for $t=n$, conditioned on the event of getting exactly $k$ samples with planted indices, we will get a
random bipartite graph with a planted $(k\times k)$-biclique (we prove the equivalence formally in Appendix \ref{sec:avg-to-dist-clique}).

One interesting approach for finding the planted clique was proposed by \citeAN{FriezeK08}. They
gave a reduction from finding a planted clique
in a random graph to finding a direction that maximizes a $2$nd order tensor norm; this was
extended to general $r$'th order tensor norm in \citeP{BV}.%\enote{added 2nd order and r'th order tensor}
 Specifically, they show that
maximizing the $r$'th moment (or the $2$-norm of an $r$'th order tensor)
allows one to recover planted cliques of size $\tilde{\Omega}(n^{1/r})$.
A related approach is to maximize the 3rd or higher moment of the distribution given by the
distributional planted clique problem. For this approach it is natural to
consider the following type of optimization algorithm: start with some
unit vector $u$, then estimate the gradient at $u$ (via samples), move
along that direction and return to the sphere; repeat to reach an approximate local
maximum.  Unfortunately, over the unit sphere, the expected $r$'th
moment function can have (exponentially) many local maxima even for
simple distributions.  A more sophisticated
approach~\citeP{Kannan-personal} is through Markov chains or simulated
annealing; it attempts to sample unit vectors from a distribution
on the sphere which is heavier on vectors that induce a higher moment,
e.g., $u$ is sampled with density proportional to $e^{f(u)}$ where
$f(u)$ is the expected $r$'th moment along $u$.  This could be
implemented by a Markov chain with a Metropolis
filter~\citeP{MetropolisRRTT53,Hastings70} ensuring a proportional steady state
distribution.  If the Markov chain were to mix rapidly, that would
give an efficient approximation algorithm because sampling from the
steady state likely gives a vector of high moment. At each step, all
one needs is to be able to estimate $f(u)$, which can be done by
sampling from the input distribution.

As we will see presently, these approaches can be easily implemented in our framework and will
have provably high complexity.
For the distributional planted biclique problem, SQ algorithms need $n^{\Omega(\log n)}$ queries to detect
planted bicliques of size $k < n^{\frac{1}{2}-\delta}$ for any $\delta > 0$. Even stronger exponential bounds apply for the more general problem of detecting planted dense subgraphs
of the same size. These bounds match the known upper bounds.
To describe these results precisely and discuss exactly what they mean for the complexity of these problems, we will need to define the models of statistical algorithms, the complexity
measures we use, and our main tool for proving lower bounds, a notion of statistical dimension of a set of distributions. We do this in the next section. In Section \ref{sec:stat-dim} we
prove our general lower bound results and in Section \ref{sec:strongclique} we estimate the statistical dimension of detecting planted bicliques and dense subgraphs.

\section{Definitions and Overview}\label{sec:defs}
Here we formally define statistical algorithms and the key notion of statistical dimension, and then describe the resulting lower bounds in detail.

\subsection{Problems over Distributions}
We begin by formally defining the class of problems addressed by our framework.
\begin{definition}[Search problems over distributions]\label{def:searchD}
  For a domain $X$,
  let $\D$ be a set of distributions over $X$, let $\F$ be a set called {\em solutions} and $\Z:\D \rightarrow 2^{\F}$   be a map from a distribution $D \in \D$ to a subset of solutions $\Z(D)
  \subseteq \F$ that are defined to be valid solutions for $D$.
  The {\em distributional search problem} $\Z$ over $\D$ and $\F$ using $t$ samples is to find a valid solution $f \in \Z(D)$ given access (to an oracle or samples from) an unknown $D \in \D$.
\end{definition}

In some settings it is natural to parameterize the set of valid solutions by additional parameters, such as accuracy. %One can also consider  average case search problems in which the choice of the input distribution is itself a random variable and the algorithm needs to succeed with high probability over the random choice.
The extension of the definition to such
settings is immediate. An example of a distributional search problem is the distributional planted $k$-biclique we described in Definition \ref{def:clique}. In this case the domain $X$ is $\{0,1\}^n$, the set of input
distributions is all the distributions with a planted $k$-biclique $\D = \{ D_S \cond S \subset [n],\ |S| = k\}$ and the set of solutions is the set of all subsets of size $k$: $\F= \{ S \cond S \subset
[n],\ |S| = k\}$. For each $D_S$ there is a single valid solution $S$. For a second example, we point the reader to the distributional MAX-XOR-SAT problem in Section \ref{sec:max-xor-sat}.

We note that this definition also captures decision problems by having $\F = \{0,1\}$.  A simple example of a
decision problem over distributions that is relevant to our discussion is that of distinguishing a planted biclique distribution from the uniform distribution over $\zo^n$ which we denote by
$U$. Here the set of input distributions is $\D = \{U\} \cup \{ D_S \cond S \subset [n],\ |S| = k\}$. The only valid solution for a planted biclique distribution $D_S$ is 1 and the only valid
solution for $U$ is 0. For a solution $f \in \F$, we denote by $\Z_f$ the set of distributions in $\D$ for which $f$ is a valid solution.

It is important to note that the number of available random samples $t$ can
have a major influence on the complexity of the problem. First, for
most problems there is a minimum $t$ for which the problem is
information-theoretically solvable. This value is often referred to as
the {\em sample complexity} of the problem.
But even for $t$ which is
larger than the sample complexity of the problem, having more samples
can make the problem easier computationally. For example, in the
context of attribute-efficient learning, there is a problem that is intractable with few
samples (under cryptographic assumptions) but is easy to solve with a
larger (but still polynomial) number of samples \citeP{Servedio:00jcss}. Our distributional
planted biclique problem exhibits the same phenomenon.

\subsection{Statistical Algorithms}
The statistical query learning model of \citeAN{Kearns:98} is a
restriction of the PAC model \citeP{Valiant84}. It introduces an oracle that allows a learning algorithm
to obtain an estimate of the expectation of any bounded function of an example. A query to such an oracle is referred to as {\em statistical query}. Kearns showed that many known PAC
learning algorithms can be expressed as algorithms using statistical queries instead of random examples themselves. The main goal of Kearns' model was to give a simple way to design
algorithms tolerant to random classification noise. Since the introduction of the model SQ algorithms have been given for many more learning tasks and the model itself
found applications in a number of other contexts such as differential privacy \citeP{BlumDMN:05,KasiviswanathanLNRS11}, learning on massively parallel architectures
\citeP{ChuKLYBNO:06} and evolvability \citeP{Feldman:08ev}.

In the same spirit, for general search problems over a distribution, we define SQ algorithms as algorithms that do
not see samples from the distribution but instead have access to a SQ oracle. The first SQ oracle we define is the natural generalization of the oracle defined by
\citeAN{Kearns:98} to samples from an arbitrary distribution.
\begin{definition}[$\STAT$ oracle]
  Let $D$ be the input distribution over the domain $X$. For a {\em tolerance} parameter $\tau > 0$, $\STAT(\tau)$ oracle is the oracle that for any query  function $h: X \rightarrow
  [-1,1]$, returns a value $$v \in \left[\E_{x\sim D}[h(x)] - \tau, \E_{x\sim D}[h(x)] + \tau\right].$$
\end{definition}

The general algorithmic techniques mentioned earlier can all be expressed as algorithms using $\STAT$ oracle instead of samples themselves, in most cases in a straightforward way. We
would also like to note that in the PAC learning model some of the algorithms, such as the Perceptron algorithm, did not initially appear to fall into the SQ framework but SQ analogues
were later found for all known learning techniques except Gaussian elimination (for specific examples,  see \citeP{Kearns:98} and \citeP{BlumFKV98}). We expect the situation to be
similar even in the broader context of search problems over distributions.

The most natural realization of $\STAT(\tau)$ oracle is one that
computes $h$ on $O(1/\tau^2)$ random samples from $D$ and returns their
average. Chernoff's bound implies that the estimate is within the desired tolerance (with constant probability). However, if $h(x)$ is very biased (e.g.~equal to 0 with high probability), it
can be estimated with fewer samples. Our primary application requires a tight bound on the number of samples necessary to solve a problem over distributions. Therefore we define a
stronger version of $\STAT$ oracle which tightly captures the accuracy of an estimate of the expectation given by random samples. %Specifically, the oracle returns the expectation to
%within the same tolerance (up to constant factors) as one expects to get from $t$ samples.
More formally, for a Boolean query function $h:X\rightarrow \{0,1\}$, $\VSTAT(t)$ can return any value $v$ for which the Binomial distribution $\Bin(t,v)$  (sum of $t$ independent Bernoulli variables with bias $v$) is statistically close (for some constant distance) to $\Bin(t,\E[h])$. See
Sec.~\ref{sec:vstat-to-sample} for more details on this correspondence.
\begin{definition}[$\VSTAT$ oracle]
  Let $D$ be the input distribution over the domain $X$. For a {\em sample size} parameter $t > 0$, $\VSTAT(t)$ oracle is the oracle that for any query function $h: X \rightarrow
  [0,1]$, returns a value $v \in \left[p - \tau, p + \tau\right],$ where $p = \E_{x\sim D}[h(x)]$ and $\tau = \max\left\{\frac{1}{t}, \sqrt{\frac{p(1-p)}{t}}\right\}$.
\end{definition}

Note that $\VSTAT(t)$ always returns the value of the expectation within $1/\sqrt{t}$. Therefore it is no weaker than
 $\STAT(1/\sqrt{t})$ and no stronger than $\STAT(1/t)$. %($\STAT$, unlike $\VSTAT$, allows non-Boolean functions but this is not a significant difference as any $[-1,1]$-valued query can be converted to a logarithmic number of $\{0,1\}$-valued queries).

The $\STAT$ and $\VSTAT$ oracles we defined can return any value
within the given tolerance and therefore can make adversarial
choices. We also aim to prove lower bounds against algorithms that use
a more benign, \sampleoracle oracle\footnote{In the STOC 2013 extended abstract, this oracle is also called the {\em unbiased} statistical oracle}. The
 \sampleoracle oracle gives the algorithm the true value of a
Boolean query function on a randomly chosen sample. This oracle is a special case of the $k$-bit sampling oracle introduced by \citeAN{Ben-DavidD98} who refer to it as the {\em weak Restricted Focus of Attention (wRFA)} model and is also equivalent to the Honest SQ oracle of \citeAN{Yang:01}. Learning in this model has been studied in more recent work motivated by communication constraints on data processing in a distributed computing system. \citeP{ZhangDJW13,SteinhardtD15,SteinhardtVW16}.
\begin{definition}[$\SAMPLE$ oracle]
  Let $D$ be the input distribution over the domain $X$. The $\SAMPLE$ oracle is the oracle that given any function $h: X \rightarrow \{0,1\}$,
  takes an independent random sample $x$ from $D$ and returns $h(x)$.
\end{definition}

Note that the $\SAMPLE$ oracle draws a fresh sample upon each time it is called.
Without re-sampling each time, the answers of the $\SAMPLE$ oracle could be easily used to recover any sample bit-by-bit, making it equivalent to having access to random
samples.
Note that the $\SAMPLE$ oracle can be used to simulate $\VSTAT$ (with
high probability) by taking the average of $O(t)$ replies of
$\SAMPLE$ for the same function $h$. While it might seem that access to
$\SAMPLE$ gives an algorithm more power than access to $\VSTAT$ we will show that $t$ samples from $\SAMPLE$ can be simulated using access to $\VSTAT(O(t))$. This will allow us
to translate our lower bounds on SQ algorithms with access to $\VSTAT$ to lower bounds against \sampleoracle algorithms.

\eat{
In the following discussion, we refer to algorithms using $\STAT$, $\VSTAT$ or $\SAMPLE$ oracles (instead of regular samples) as {\em statistical algorithms}. %Algorithms using the $\SAMPLE$ oracle are called {\em \sampleoracle} algorithms.
The {\em query complexity} of a statistical algorithm is defined to be the number of calls it makes to the oracle.
}

\subsection{Statistical Dimension}
The main tool in our analysis is an information-theoretic bound on the
complexity of statistical algorithms.
Our definitions originate from the statistical query
(SQ) dimension \citeP{BlumFJKMR94} used to characterize SQ learning algorithms. Roughly speaking, the SQ dimension corresponds to the number of nearly
uncorrelated labeling functions in
a class (see Section \ref{sec:relate-2-sq} for the details of the definition
and the relationship to our bounds).

We introduce a natural generalization and strengthening of this approach to search problems over arbitrary sets of distributions and prove lower bounds on the complexity of statistical
algorithms based on the generalized notion. Our definition departs from SQ dimension in three aspects. (1) Our notion applies to any set of distributions; in the learning setting
all known definitions of statistical dimension require fixing the distribution over the domain and only allow varying the labeling function. Such an extension was not known even in the context of
PAC learning. (2) Instead of relying on a bound on pairwise correlations, our dimension relies on a bound on average correlations in a large set of distributions. This weaker condition
allows us to derive tight bounds on the complexity of SQ algorithms for the planted $k$-biclique problem. (3) We show that our notion of dimension also gives lower bounds for the
stronger $\VSTAT$ oracle (without incurring a quadratic loss in the parameter).

We now define our dimension formally.
For two functions $f,g: X \rightarrow \R$ and a distribution $D$ with
probability density function $D(x)$, the inner product of $f$ and $g$
over $D$ is defined as
\begin{align*}
\langle f, g \rangle_D \doteq \E_{x \sim D}[f(x)g(x)].
\end{align*}
The norm of $f$ over $D$ is $\|f\|_D = \sqrt{\langle f, f
  \rangle_D}$. We remark that, by convention, the integral from the
inner product is taken only over the support of $D$, i.e. for $x\in X$
such that $D(x)\not=0$.
Given a distribution $D$ over $X$ let $D(x)$ denote the probability
 density function of $D$ relative to some fixed underlying measure
 over $X$ (for example uniform distribution for discrete $X$ or
 Lebesgue measure over $\R^n$). Our bound is based on the inner
 products between functions of the following form: $(D'(x)-D(x))/D(x)$
 where $D'$ and $D$ are distributions over $X$.  For this to be
 well-defined, we will only consider cases where $D(x)=0$ implies
 $D'(x)=0$, in which case $D'(x)/D(x)$ is treated as 1. To see why
 such functions are relevant to our discussion, note that for  every
 real-valued function $f$ over $X$,
\begin{align*}
  \E_{x \sim D'}[f(x)] - \E_{x \sim D}[f(x)] & = \E_{x \sim
    D}\left[\frac{D'(x)}{D(x)}f(x)\right] -
  \E_{x \sim D}[f(x)]   = \left\langle \frac{D'-D}{D},f\right\rangle_D.
\end{align*}
This means that the inner product of any function $f$ with $(D'-D)/D$ is equal to the difference of expectations of $f$ under the two distributions. Analyzing this quantity for an arbitrary set of functions $f$ was the high-level approach of statistical query lower bounds for learning. Here we depart from this approach, by defining a {\em pairwise correlation} of two distributions, independent of any specific query function. For two distributions $D_1, D_2$ and a reference distribution $D$, their pairwise correlation is defined as:
\[
\chi_D(D_1, D_2) = \left|\left\langle
    \frac{D_1}{D}-1, \frac{D_2}{D}-1 \right\rangle_D \right|.
\]
When $D_1 = D_2$, the quantity
$\langle \frac{D_1}{D}-1, \frac{D_1}{D}-1\rangle_D $ is known as the $\chi^2(D_1, D)$ distance and is widely used for hypothesis testing in statistics \citeP{Pearson:1900}.

A key notion for
our statistical dimension is the {\em average correlation} of a set of distributions $\D'$ relative to a distribution $D$. We denote it by $\rho(\D',D)$ and define as follows:
\begin{align*}
\rho(\D',D) \doteq \frac{1}{|\D'|^2}\sum_{D_1,D_2 \in \D'} \chi_D(D_1, D_2) =
\frac{1}{|\D'|^2}\sum_{D_1,D_2 \in \D'} \left|\left\langle
    \frac{D_1}{D}-1, \frac{D_2}{D}-1 \right\rangle_D \right| .
\end{align*}

%Known lower bounds for SQ learning are based on bounding the pairwise correlations between functions.
Bounds on pairwise correlations easily imply bounds on
the average correlation (see Lemma \ref{lem:cor-2-sda} for a proof). In Section \ref{sec:pairwise-correlations} we describe a pairwise-correlation version of our bounds. It is sufficient for some applications and generalizes the
statistical query dimension from learning theory (see Section \ref{sec:relate-2-sq} for the details).
However, to obtain our nearly tight lower bounds for planted biclique, we will need to bound the average pairwise correlation directly, and with significantly better bounds than what is possible from pairwise correlations alone.

We are now ready to define the concept of statistical dimension.
We first define the statistical dimension with average correlation of a set of distributions relative to some reference distribution. It captures the complexity of distinguishing distributions in $\D$ from $D$.
\begin{definition}\label{def:sdima-set}
  For $\bar{\gamma}>0$, domain $X$, a set of distributions $\D$ over $X$ and a reference distribution $D$ over $X$
  the \textbf{statistical dimension} of $\D$ relative to $D$ with average correlation $\bar{\gamma}$ is defined to be
   the largest value $d$ such that for any subset $\D' \subseteq \D$, where $|\D'| \ge |\D|/d$, $\rho(\D',D) \leq \bar{\gamma}$. We denote it by
   $\SDA(\D,D,\bar{\gamma})$.
\end{definition}
Intuitively, the definition says that any $1/d$ fraction of the set of distributions has low pairwise correlation; the largest such $d$ is the statistical dimension.

For general search problems over distributions we define the statistical dimension by reducing it to the statistical dimension of some set of input distributions relative to some reference distribution.
\begin{definition}\label{def:sdima}
  For $\bar{\gamma}>0$, domain $X$, a search problem $\Z$ over a set of solutions $\F$
  and a class of distributions $\D$ over $X$, let $d$ be the largest value
  such that there exists a {\em  reference} distribution $D$ over $X$ and a finite set of distributions $\D_D \subseteq \D$ with the following property:
  for any solution $f\in \F$ the set $\D_f = \D_D \setminus \Z_f$ is non-empty and
  $\SDA(\D_f,D,\bar{\gamma})\geq d$.
   We define the \textbf{statistical dimension} with average correlation $\bar{\gamma}$ of $\Z$
   to be $d$ and denote it by $\SDA(\Z,\bar{\gamma})$.
\end{definition}
The statistical dimension with average correlation $\bar{\gamma}$ of a search problem over distributions gives
a lower bound on the complexity of any deterministic statistical algorithm for the
problem that uses queries to $\VSTAT(1/(3\bar{\gamma}))$.
\begin{theorem}\label{thm:avgvstat}
  Let $X$ be a domain and $\Z$ be a search problem over a set of solutions $\F$
  and a class of distributions $\D$ over $X$. For $\bar{\gamma} > 0$ let $d = \SDA(\Z,\bar{\gamma})$.
  Any SQ algorithm requires at least $d$ calls to $\VSTAT(1/(3\bar{\gamma}))$ oracle to solve $\Z$.
\end{theorem}
In Section \ref{sec:lower-bound-vstat} we give a refinement of $\SDA$, by introducing a parameter which additionally bounds the size of the set $\D_f$ (and not just that it is non-empty). This refined notion
allows us to extend the lower bound to randomized SQ algorithms.
In Section \ref{sec:vstat-to-sample} we use this refined notion of $\SDA$ to also show that (with high probability) one can simulate $t$ samples of $\SAMPLE$ using
$\VSTAT(O(t))$. This implies that lower bounds on $\SDA$ imply lower bounds on the number of queries required by any \sampleoracle algorithm (Theorem \ref{thm:avgsample-v}).

In Section \ref{sec:app2SQ} we show that our bounds generalize and strengthen the known results for SQ learning that are based on SQ-DIM \citeP{BlumFJKMR94,Yang05}. In the statement below, the statistical dimension $\SDA(\C, D', \bar{\gamma})$ uses the average pairwise correlation of Boolean functions from a set $\C$ relative to a distribution $D'$ over a domain $X'$, that is $\la f_1,f_2\ra_{D'}$, where $f_1,f_2\in \C$ (rather than distributions as in the definitions above). It is formally defined in Section \ref{sec:app2SQ} and is always at least as large as the statistical query dimension used in earlier work in learning theory.

\begin{theorem}
\label{thm:general-learning}
Let $\C$ be a set of Boolean functions, $D'$ be a distribution over $X'$ and let $d=\SDA(\C,D',\bar{\gamma})$ for some $\bar{\gamma} > 0$. Then any SQ algorithm that, with probability at least $2/3$, learns $\C$ over $D'$ with error $\eps < 1/2-\sqrt{1/(3\bar{\gamma})} $ requires at least $d/3-1$ queries to $\VSTAT(1/(3\bar{\gamma}))$.
\end{theorem}

At a high level, our proof works as follows. The first step of the proof is a reduction from a decision problem in which the algorithm only needs to distinguish all the distributions in the set
$\D_D$ (except those in $\Z_f$ for some $f$) from the reference distribution $D$. To distinguish between distributions the algorithm needs to ask a query $g$ such that
$\E_D[g]$ cannot be used as a response of $\VSTAT(1/(3\bar{\gamma}))$ for $D' \in \D_f$. In the key component of the proof we show that if a query function $g$ to
$\VSTAT(1/(3\bar{\gamma}))$ distinguishes between a distribution $D$ and any distribution $D' \in \D'$, then $\D'$ must have average correlation of at least $\bar{\gamma}$ relative
to $D$. The condition that for any $|\D'| \ge |\D_f|/d$, $\rho(\D',D) < \bar{\gamma}$ then immediately implies that at least $d$ queries are required to distinguish any distribution in
$\D_f$ from $D$. We remark that an immediate corollary of this proof technique is that the decision problem in which the algorithm needs to decide whether the input distribution is in
$\D_f$ or is equal to the reference distribution $D$ also has statistical dimension at least $d$. We elaborate on this in Theorem \ref{thm:avgvstat-random-decision} where we give a
simplified version of our lower bound for decision problems of this type.

\subsection{Applications to the Planted Biclique Problem}
We prove the following lower bound for the distributional planted biclique problem.% with both statistical oracles ($\VSTAT$ and $\SAMPLE$) .
\begin{theorem}\label{thm:planted-stat}
  For any constant $\delta > 0$, any $k \leq n^{1/2 - \delta}$ and $r > 0$, at least
  $n^{\Omega(\log{r})}$ queries to $\VSTAT(n^2/(rk^2))$ are required to solve the distributional planted $k$-biclique with probability at least $2/3$. In particular, no
  polynomial-time statistical algorithm can solve the problem using queries to $\VSTAT(o(n^2/k^2))$ and any SQ algorithm requires $n^{\Omega(\log{n})}$ queries to
  $\VSTAT(n^{2-\delta}/k^2)$. This lower bound also applies to the problem of distinguishing any planted $k$-biclique distribution from the uniform distribution over $\zo^n$ (no planting).
\end{theorem}

This bound is essentially tight. For every index in the planted set $S$, the probability that the corresponding bit of a randomly chosen point is set to $1$ is $1/2 + k/(2n)$, whereas for
every index not in $S$, this probability is $1/2$. Therefore using $n$ queries to $\VSTAT(16n^2/k^2)$ (i.e., of tolerance $k/4n$) it is easy to recover $S$. Indeed, this can be done by using the query functions $h_i(x)=x_i$, for each $i\in [n]$. So, the answers of the $\VSTAT$ oracle represent the
expected value of the $i$th bit over the sample.

There is also a SQ algorithm that uses $n^{O(\log{n})}$ queries to $\VSTAT(25n/k)$ (corresponding to a significantly smaller number of samples) to find the planted set for any $k \geq \log n$. In fact, the
same algorithm can be used for the standard planted clique problem that achieves complexity $n^{O(\log n)}$. We enumerate over subsets $T\subseteq [n]$ of $\log n$ indices and
query $\VSTAT(25n/k)$ with the function $g_T: \{0,1\}^n\rightarrow \{0,1\}$ defined as $1$ if and only if the point has ones in all coordinates in $T$.
Therefore, if the set $T$ is included in the planted set then
\[
\E_D[g_T] = \frac{k}{n} \cdot 1 + \left(1- \frac{k}{n}\right)2^{-\log n} \in \left[\frac{k}{n}, \frac{k+1}{n}\right].
\]
With this expectation, $\VSTAT(25n/k)$ has tolerance at most $\sqrt{k(k+1)/25n^2} \le (k+1)/5n$ and
will return at least $k/n - (k+1)/(5n) > 3k/(4n)$. If, on the other hand, at least one element of $T$ is not from the planted set, then $\E_D[g_T] \leq k/(2n) + 1/n$ and
$\VSTAT(25n/k)$ will return at most $(k+2)/(2n) + (k+2)/(5n) < 3k/(4n)$. Thus, we will know all
$(\log n)$-sized subsets of the planted set and hence the entire planted set. We remark that this algorithm demonstrates the difference between $\STAT$ and $\VSTAT$ oracles.
Implementing this algorithm using the $\STAT$ oracle would require  tolerance of $\Omega(k/n)$ which corresponds to $O(n^2/k^2)$ samples. This is the same tolerance as the
polynomial-time degree-based algorithm needs (estimate degree of each vertex), so one cannot hope to have a superpolynomial lower bound against $\STAT(k/n)$.

To summarize, $n$ samples directly correspond to having access to
$\VSTAT(O(n))$. The discussion above shows that the distributional planted biclique problem can be  solved in
polynomial time when $k=\Omega(\sqrt{n})$. At the same time, Theorem \ref{thm:planted-stat}  implies that for $k \leq
n^{1/2 - \delta}$, any SQ algorithm will require
$n^{\Omega(\log{n})}$ queries to $\VSTAT(n^{1+\delta})$.

We now turn to stating our bounds for \sampleoracle algorithms.
\begin{theorem}\label{thm:planted-strong}
  For any constant $\delta > 0$ and any $k\leq n^{1/2 - \delta}$,
  any \sampleoracle algorithm that with probability at least $2/3$ can distinguish between the uniform distribution and any planted $k$-biclique distribution requires
  $\Omega(n^2/k^2)$ queries to $\SAMPLE$.
\end{theorem}

Each query of a \sampleoracle algorithm uses a new sample from $D$. Therefore this bound implies that any algorithm that does not reuse samples will require
$\Omega(n^2/k^2)$ samples. To place this bound in context, we note that it is easy to detect whether a biclique of size $k$ has been planted using $\tilde{O}(n^2/k^2)$ samples
(as before, to detect if a coordinate $i$ is in the planted set we can compute the average of $x_i$ on $\tilde{O}(n^2/k^2)$ samples). Of course, finding all coordinates in the set would require
reusing samples (which \sampleoracle algorithms cannot do). Note that $n^2/k^2 \leq n$ if and only if $k \geq \sqrt{n}$.

A closely related problem is the {\em planted densest subgraph} problem,
where edges in the planted subset appear with higher probability than
in the remaining graph. This is a variant of the densest $k$-subgraph
problem, which itself is a natural generalization of $k$-clique that
asks to recover the densest $k$-vertex subgraph of a given $n$-vertex
graph
\citeP{Feige02,Khot04a,BhaskaraCCFV10,BhaskaraCVGZ12}.  The conjectured
hardness of its average case variant, the planted densest subgraph
problem, has been used in public key encryption schemes
\citeP{ApplebaumBW10} and in analyzing parameters specific to
financial markets \citeP{AroraBBG10}.   We define the following distributional version of this problem:
\begin{problem} \label{def:densegraph} Fix $0 < q < p \le 1$. For $1
  \le k \le n$, let $S\subseteq [n]$ be a set of $k$
  vertex indices and $D_S$ be a distribution over $\{0,1\}^n$ such
  that when $x \sim D_S$, with probability $1-(k/n)$ the entries of
  $x$ are independently $q$-biased Bernoulli
  variables, and with probability $k/n$ the $k$ coordinates in $S$ are
  independently chosen $p$-biased Bernoulli variables,
  and the rest are independently chosen $q$-biased
  Bernoulli variables. The \textbf{distributional $(p,q)$-planted
    densest $k$-subgraph} problem is to find the unknown subset $S$
  given access to samples from $D_S$.
\end{problem}
Our approach and lower bounds extend in a straightforward manner to this problem. In Section \ref{sec:densest} we analyze this general setting and give lower bounds for all settings of $p$ and $q$. Here we describe a special case of our lower bounds when $q=1/2$ and $p=1/2+\alpha$. Our lower bound becomes exponential as $\alpha$ becomes (inverse-polynomially) close to $0$. Specifically:
\begin{corollary}\label{thm:densest-stat-intro}
  For any constant $\delta > 0$, any $k \leq n^{1/2 - \delta}$, $\alpha > 0$, $\ell
  \leq \min\{k,1/(4\alpha^2)\}$, at least
  $n^{\Omega(\ell)}$ queries to $\VSTAT(n^2/(48\ell \alpha^2 k^2))$ are
  required to solve the distributional $(1/2+\alpha,1/2)$-planted
    densest $k$-subgraph with probability at least $2/3$.
\end{corollary}
For example, consider the setting $k=\ell=n^{1/3}$ and $\alpha = n^{-1/4}$. It is not hard to see that for this setting the problem can be solved on a random bipartite graph with $n$ vertices on both sides (in exponential time). Our lower bound for this setting implies that at least $n^{\Omega(n^{1/3})}$ queries to $\VSTAT(n^{3/2})$ will be required.
 Additional corollaries for the distributional $(p,q)$-planted densest $k$-subgraph can be found in Section \ref{sec:densest}.

\paragraph{Relation to the planted $k$-biclique problem}
The upper and lower bounds we described for statistical algorithms match the state of the art for the average-case planted $k$-biclique and planted $k$-clique problems. Moreover,
our lower bounds for the distributional versions of the planted $k$-biclique problem have implications for the hardness of the average-case planted $k$-biclique problem.  An instance of the latter problem is a random $n\times n$ bipartite graph with a $k \times k$ biclique planted randomly. In Appendix \ref{sec:avg-to-dist-clique}, we show that the average-case planted $k$-biclique is
equivalent to our distributional planted $k$-biclique with $n$ samples. Specifically, a single sample corresponds to the adjacency list of a vertex on the left, and $n$ samples correspond
to the adjacency matrix of the bipartite graph. By this equivalence, an algorithm that solves the average-case planted bipartite $k$-clique problem will also solve the distributional
planted $k$-biclique with $n$ samples. Our lower bounds for the distributional problem therefore imply that the planted $k$-biclique problem would require a non-statistical approach, i.e., one for which there is no statistical analogue.% (that uses $\VSTAT(O(n))$ oracle).

\eat{
Before describing the lower bounds derived using our approach we address the following
 natural question: what do superpolynomial statistical lower bounds for a problem say about the performance of any given algorithm that only uses summations of functions over
 samples? If the algorithm uses each sample only for a single function evaluation then our lower bounds against unbiased statistical algorithms give a lower bound on the sample
 complexity (and hence running time) of any such algorithm for the problem. In most cases, however, such an algorithm is not statistical in the formal sense since it is likely using its
 samples more than once and is not based on oracles. Therefore our bounds do not constitute a proof that such an algorithm fails to solve the problem. At the same time they indicate
 that a proof of correctness or other formal performance analysis which relies on concentration of the sums used by the algorithm (as is almost always the case) will not be possible. They
 also imply that the algorithm is not robust to even tiny perturbations of the values of sums (which can arise from noise for example). Both of these points give strong evidence that the
 algorithm is unlikely to be successful on all input distributions.
}

\subsection{Subsequent Work}
In subsequent work, Feldman, Perkins and Vempala \citeY{FeldmanPV:13} introduced a  notion of statistical dimension that is based on the spectral norm of the correlation matrix of large sets of
distributions. It is always at least as large as the average correlation-based dimension defined here and also leads to lower bounds on the complexity of SQ
algorithms using $\VSTAT$.   % since for every $m \times n$ matrix $A$,$ \|A\|_2 \leq \frac{1}{mn}\sum_ij |A_ij|$.
Using this dimension they proved tight lower bounds on the complexity of statistical algorithms for planted $k$-SAT and Goldreich's  pseudo-random generator. In addition, they described
statistical algorithms based on power iteration with nearly matching upper bounds.
Finally, they demonstrate that lower bounds against SQ algorithms can be used to derive concrete lower bounds for convex relaxations of the problem.

\citeAN{FeldmanPV:13} have also extended the lower bounds against $\SAMPLE$ to lower bounds against the $k$-bit version of $\SAMPLE$ at the expense of factor $2^k$ blow-up in the number of queries. Steinhardt, Valiant and Wager \citeY{SteinhardtVW16} gave a more direct approach for proving lower bounds against this oracle that is closely related to the techniques here and in \citeAN{FeldmanPV:13}. They have further showed that statistical queries can be used to simulate the oracle that that extracts $k$ bits from each sample in an interactive way (rather than at once).

Building on our approach, \citeAN{Feldman:16sqd} described new notions of statistical dimension and proved that they tightly characterize the SQ complexity of solving general search problems over distributions for both $\STAT$ and $\VSTAT$ oracle. He also simplified the analysis of $\VSTAT(t)$ by showing that it is equivalent (up to constant factors) to returning any value $v$ such that $|\sqrt{v} - \sqrt{\E_D[h]}|\leq 1/\sqrt{t}$.
Some additional recent applications of SQ lower bounds that are related to our work include learning of the Ising model \citeP{BreslerGS14a}, convex optimization \citeP{FeldmanGV:15} and distribution-independent PAC learning of lines over finite fields \citeP{Feldman:16sqd}.

The distributional planted $k$-biclique problem introduced here is a simple and natural problem that shows a remarkable property: information-theoretically it can be solved with many fewer samples
than is necessary for any known efficient algorithm (and no efficient statistical algorithm exists). In particular, any algorithm that solves our problem with less than $n$ samples will also
solve the average-case $k$-biclique problem (that is at least as hard as the usual planted $k$-clique problem). In several more recent works, reductions from the planted clique
problem were used to demonstrate a similar phenomenon in a number of important problems in statistics and machine learning
\citeP{BerthetR:13,MaW13a,HajekWX15,GaoMZ14,WangBS15,CaiLR15}.

\section{Lower Bounds from Statistical Dimension}\label{sec:stat-dim}
In this section we prove the general lower bounds. In later sections, we will compute the parameters in these bounds for specific problems of interest.
\subsection{Lower Bounds for Statistical Query Algorithms}
\label{sec:lower-bound-vstat}
We start by proving Theorem \ref{thm:avgvstat} which is the basis of all our lower bounds. In fact, we will prove a stronger version of this theorem which also applies to randomized
algorithms. For this version we need an additional parameter in the definition of SDA.
\begin{definition}\label{def:sdima-random}
  For $\bar{\gamma}>0$, $\eta > 0$, domain $X$ and a search problem $\Z$ over a set of solutions $\F$
  and a class of distributions $\D$ over $X$, let $d$ be the largest value
  such that there exists a {\em  reference} distribution $D$ over $X$ and a finite set of distributions $\D_D \subseteq \D$ with the following property:
  for any solution $f\in \F$ the set $\D_f = \D_D \setminus \Z_f$ has size at least $(1-\eta) \cdot |\D_D|$ and $\SDA(\D_f,D,\bar{\gamma}) \geq d$. We define the statistical dimension with average correlation $\bar{\gamma}$ and solution set bound $\eta$ of $\Z$ to be $d$ and denote it by
$\SDA(\Z,\bar{\gamma},\eta)$.
\end{definition}
Note that for any $\eta < 1$, $\SDA(\Z,\bar{\gamma}) \geq \SDA(\Z,\bar{\gamma},\eta)$ and for $\eta = 1-1/|\D_D|$, we get $\SDA(\Z,\bar{\gamma}) =
\SDA(\Z,\bar{\gamma},\eta)$,
where $\D_D$ is the set of distributions that maximizes $\SDA(\Z,\bar{\gamma})$.
\begin{theorem}
\label{thm:avgvstat-random}
  Let $X$ be a domain and $\Z$ be a search problem over a set of solutions $\F$
  and a class of distributions $\D$ over $X$. For $\bar{\gamma} > 0$ and $\eta \in (0,1)$ let $d = \SDA(\Z,\bar{\gamma},\eta)$.
  Any randomized SQ algorithm that solves $\Z$ with probability $\alpha > \eta$ requires at least $\frac{\alpha -\eta}{1-\eta}  d$ calls to $\VSTAT(1/(3\bar{\gamma}))$.
\end{theorem}

Theorem \ref{thm:avgvstat} is obtained from Theorem \ref{thm:avgvstat-random} by setting $\alpha = 1$ and using any $1-1/|\D_D| \leq \eta < 1$. Further, for any $\eta < 1$,
$\SDA(\Z,\bar{\gamma}) \geq \SDA(\Z,\bar{\gamma},\eta)$ and therefore for any $\eta < 1$, a bound on $\SDA(\Z,\bar{\gamma},\eta)$ can be used in Theorem \ref{thm:avgvstat}
in place of bound on $\SDA(\Z,\bar{\gamma})$. We now prove Theorem \ref{thm:avgvstat-random}.
\begin{proof}[of Theorem \ref{thm:avgvstat-random}]
  We prove our lower bound by exhibiting a distribution over inputs (which are distributions over $X$) for which every deterministic SQ algorithm that solves $\Z$ with probability
  $\alpha$ (over the choice of input) requires at least $(\alpha-\eta) \cdot d/(1-\eta)$ calls to $\VSTAT(1/(3\bar{\gamma}))$. The claim of the theorem will then follow by Yao's minimax
  principle \citeP{Yao:1977}.

 Using the notation of Definition \ref{def:sdima-random}, let $D$ be the reference distribution and $\D_D$ be a set of distributions for which the value $d$ is achieved.  Let $\A$ be a deterministic SQ algorithm that uses $q$ queries to
  $\VSTAT(1/(3\bar{\gamma}))$ to solve $\Z$ with probability $\alpha$ over the  random and uniform choice of a distribution from $\D_D$.
  Consider the execution of $\A$ in which to each query $h$ of $\A$, the oracle returns exactly $\E_D[h]$ and let $f$ denote the output.
  Let the set $\D_D^+ \subseteq \D_D$ be the set of
  distributions on which $\A$ is successful for all valid responses of $\VSTAT(1/(3\bar{\gamma}))$.
   Let $\D^+ =\D_f \cap \D_D^+$ (recall that we defined $\D_f = \D_D \setminus \Z_f$). We observe that $\D^+ =\D_D^+ \setminus (\D_D \setminus \D_f)$ and therefore
  \equ{|\D^+| \geq |\D_D^+| - |\D_D \setminus \D_f| \geq \alpha |\D_D| - |\D_D \setminus \D_f| = \frac{\alpha |\D_D| - |\D_D \setminus \D_f|}{|\D_D| - |\D_D \setminus \D_f|}
  |\D_f| \geq \frac{\alpha-\eta}{1-\eta} |\D_f| \label{eq:success-bound}.}

  By the definition of $\SDA(\Z,\bar{\gamma})$, it holds that $\SDA(\D_f,D,\bar{\gamma}) \geq d$. In Lemma \ref{lem:hard-to-distinguish} given below, we will show that under the conditions of this proof, $\SDA(\D_f,D,\bar{\gamma}) \geq d$ implies that $\A$ must use at least $q \geq d|\D^+|/|\D_f|$ queries. By inequality (\ref{eq:success-bound}), $q \geq \frac{\alpha-\eta}{1-\eta} \cdot d$ giving the desired lower bound.

\end{proof}

The proof of Theorem \ref{thm:avgvstat-random} relies on the following lemma that translates a lower bound on $\SDA(\D_f,D,\bar{\gamma})$ into a lower bound on the number of queries that $\A$ needs to use. Its proof is based on ideas from \citeP{Szorenyi09} and \citeP{Feldman:12jcss}.
\begin{lemma}
\label{lem:hard-to-distinguish}
  Let $X$ be a domain and $\Z$ be a search problem over a set of solutions $\F$
  and a class of distributions $\D$ over $X$. Let $\A$ be a (deterministic) SQ algorithm for $\Z$ that uses at most $q$ queries to $\VSTAT(1/(3\bar{\gamma}))$. For a distribution $D$, consider the execution of $\A$ on $D$ in which to each query $h$ of $\A$, the oracle returns exactly $\E_D[h]$ and let $f$ denote the output. For a set of distributions $\D_f \subseteq \D \setminus\Z_f$ and $\bar{\gamma} > 0$, let $d= \SDA(\D_f,D,\bar{\gamma})$. Let $\D^+$ be the set of all distributions in $\D_f$ for which $\A$ successfully solves $\Z$ for all valid responses of $\VSTAT(1/(3\bar{\gamma}))$. Then $q \geq d\cdot|\D^+|/|\D_f|$.
\end{lemma}
\begin{proof}
   Let $h_1,h_2,\ldots,h_q$  be the queries asked by $\A$ when executed on $D$ with the exact responses of the oracle. Let  $m= |\D^+|$ and we denote the distributions in $\D^+$
  by $\{D_1,D_2,\ldots,D_m\}$.  For every $k \leq q$, let
  $A_k$ be the set of all distributions $D_i$ such
  that $$\left|\E_{D}[h_k(x)] - \E_{D_i}[h_k(x)]\right| > \tau_{i,k}
  \doteq \max\left\{\frac{1}{t},
    \sqrt{\frac{p_{i,k}(1-p_{i,k})}{t}}\right\},$$ where we use $t$ to
  denote $1/(3\bar{\gamma})$ and $p_{i,k}$ to denote
  $\E_{D_i}[h_k(x)].$ To prove the desired bound we first prove the
  following two claims:
  \begin{enumerate}
  \item $\sum_{k\leq q}|A_k| \geq m$;
  \item for every $k\leq q$, $|A_k| \leq |\D_f|/d$.
  \end{enumerate}
  Combining these two implies that $q |\D_f|/d \geq m$ or, equivalently, $q \geq d|\D^+|/|\D_f|$.

  In the rest of the proof for conciseness we drop the subscript $D$
  from inner products and norms. To prove the first claim we assume, for the sake of contradiction, that there exists $D_i \not\in \cup_{k\leq q} A_k$. Then for every $k\leq q$,
  $|\E_{D}[h_k(x)] - \E_{D_i}[h_k(x)]| \leq \tau_{i,k}$. This implies that $\E_{D}[h_k(x)]$ are within $\tau_{i,k}$ of $\E_{D_i}[h_k(x)]$. By the
  definition of $\VSTAT(t)$, this implies that the responses we used in our execution of $\A$ on $D$ are also valid responses of $\VSTAT(t)$ when $\A$ is executed on $D_i$. The output of this execution is $f$ and hence it must be a valid solution for $D_i$. This contradicts the definition of $\D^+$ since it is a subset of $\D_f\subseteq \D \setminus \Z_f$.

To prove the second claim, suppose that for some $k\in [d]$, $|A_k| > |\D_f|/d$. Let $p_k = \E_{D}[h_k(x)]$ and assume that $p_k \leq 1/2$ (when $p_k>1/2$ we just replace $h_k$
by $1-h_k$ in the analysis below). First we note that:
\begin{align*}
  \E_{D_i}[h_k(x)] - \E_{D}[h_k(x)] = \E_D~\left[~\frac{D_i(x)}{D(x)}~
    h_k(x)\right]-\E_{D}[h_k(x)]= \left\langle h_k,
    \frac{D_i}{D}-1\right\rangle = p_{i,k}-p_k.
\end{align*}
  Let $\hat{D}_i(x)=\frac{D_i(x)}{D(x)}-1$, (where the convention is
  that $\hat{D}_i(x)=0$ if $D(x)=0$).  We will next show upper and
  lower bounds on the following quantity
$$ \Phi = \left\la h_k, \sum_{D_i\in A_k}  \hat{D}_i\cdot \sgn \langle h_k, \hat{D}_i\rangle \right\rangle.$$
By Cauchy-Schwartz we have that
\begin{eqnarray}
  \Phi^2 = \left\la h_k, \sum_{D_i\in A_k}  \hat{D}_i\cdot \sgn \langle h_k, \hat{D}_i \rangle \right\rangle^2
\nonumber
&\leq& \| h_k\|^2 \cdot \left\|\sum_{D_i \in A_k}  \hat{D}_i\cdot \sgn \langle h_k, \hat{D}_i\rangle\right\|^2  \\
&\leq& \| h_k\|^2 \cdot \left( \sum_{D_i,D_j \in A_k} \left| \la \hat{D}_i, \hat{D}_j \ra \right| \right) \nonumber \\
&\leq& \| h_k\|^2 \cdot \rho(A_k,D) \cdot |A_k|^2 . \label{eq2ndm-a}
\end{eqnarray}

We also have that
\begin{eqnarray}
  \Phi^2 = \left\la h_k, \sum_{D_i\in A_k}  \hat{D}_i\cdot \sgn \langle h_k, \hat{D}_i\rangle \right\rangle^2&=& \left(\sum_{D_i\in A_k} \la h_k, \hat{D}_i \ra \cdot \sgn \langle
  h_k, \hat{D}_i \ra\right)^2 \nonumber \\
  &\ge& \left(\sum_{D_i\in A_k} |p_{i,k}-p_k| \right)^2. \label{eq:tausq-a}
\end{eqnarray}
To evaluate the last term of this inequality we use the fact that $|p_{i,k}-p_k| \geq \tau_{i,k} = \max\{1/t, \sqrt{p_{i,k}(1-p_{i,k})/t}\}$.
Next we use a simple fact (proved in Lemma \ref{lem:flip-probabilily} below) that $|p_{i,k}-p_k| \geq \max\{1/t, \sqrt{p_{i,k}(1-p_{i,k})/t}\}$ implies that $|p_{k}-p_{i,k}| \geq \sqrt{\frac{\min\{p_k,1-p_k\}}{3t}}$ to obtain: For every $D_i \in A_k$, \equ{|p_{k}-p_{i,k}| \geq \sqrt{\frac{\min\{p_k,1-p_k\}}{3t}} = \sqrt{\frac{p_{k}}{3t}}. \label{eq:pk-diff}}
By substituting equation (\ref{eq:pk-diff}) into (\ref{eq:tausq-a}) we get that $\Phi^2 \geq \frac{p_{k}}{3t} \cdot |A_k|^2$.

We note that, $h_k$ is a $[0,1]$-valued function and therefore $\|h_k\|^2 \leq p_k$. Substituting this into equation (\ref{eq2ndm-a}) we get that $\Phi^2 \leq p_k \cdot \rho(A_k,D)
\cdot |A_k|^2$.  By combining these two bounds on $\Phi^2$ we obtain that $\rho(A_k,D) \geq 1/(3t) = \bar{\gamma}$ which
contradicts the definition of $\SDA$.
\end{proof}
\begin{remark}
\label{rem:lowerbound-stat}
We remark that for algorithms using the $\STAT$ oracle, the proof can be simplified somewhat.
For $\tau = \sqrt{\bar{\gamma}}$,  $$\Phi^2
\geq \left(\sum_{D_i\in A_k} |p_{i,k}-p_k| \right)^2 \geq \tau^2
|A_k|^2$$ and the proof could be obtained by directly combining
equations (\ref{eq2ndm-a}) and (\ref{eq:tausq-a}) to get a
contradiction. This also eliminates the factor of $3$ in the
bound and the assumption that queries are $[0,1]$-valued can be relaxed to $[-1,1]$-valued queries since it suffices that
$\|h_k\|^2 \leq 1$. This leads to an identical lower bound on the number of queries for $\STAT(\sqrt{\bar{\gamma}})$ in place of $\VSTAT(1/(3\bar{\gamma}))$.
\end{remark}

We now prove a bound on the distance between any $p\in [0,1]$ and $p'$ which is returned by $\VSTAT(t)$ on a query with expectation
$p$ in terms of $p'$ that we used in the proof of Lemma~\ref{lem:hard-to-distinguish}.
\begin{lemma}
\label{lem:flip-probabilily}
For an integer $t$ and any $p \in [0,1]$, let $p' \in [0,1]$ be such that $|p'-p| \geq \tau = \max\left\{\frac{1}{t},
    \sqrt{\frac{p(1-p)}{t}}\right\}$. Then $|p'-p| \geq \sqrt{\frac{\min\{p',1-p'\}}{3t}}$.
\end{lemma}
\begin{proof}
First note that our conditions and bounds do not change if we replace both $p$ and $p'$ with $1-p$ and $1-p'$, respectively. Therefore it is sufficient to prove the bound when $p \leq
1/2$.
We know that $|p'-p| \geq \tau = \max\{1/t, \sqrt{p(1-p)/t}\}$.
If $p \geq 2p'/3$ then certainly $$|p'-p| \geq \sqrt{\frac{p(1-p)}{t}} \geq \sqrt{\frac{\frac{2}{3}p' \cdot \frac{1}{2}}{t}} = \sqrt{\frac{p'}{3t}}.$$ Otherwise (when $p < 2p'/3$),
$p' - p \geq p' - 2p'/3 = p'/3$.
We also know that $|p-p'| \geq \tau \geq 1/t$ and therefore $|p-p'| \geq \sqrt{\frac{p'}{3t}}$.
\end{proof}

\subsubsection{Decision Problems}
For decision problems, our dimension and lower bounds can be simplified. We denote by $\B(\D,D)$ a decision problem in which the input distribution $D'$ either equals $D$ or belongs to $\D$ and the goal of the algorithm is to identify whether $D' = D$ or $D'\in \D$. For example, for the distributional planted $k$-biclique problem, the decision version is to determine whether the given input distribution corresponds to a planted $k$-biclique or to one with no planting (uniform distribution on $\{0,1\}^n$.

For the decision problem $\B(\D,D)$ our notion of dimension simplifies to the following.
\begin{definition}\label{def:sdim-decision}
  For $\bar{\gamma}>0$, domain $X$ and a decision problem $\B(\D,D)$, let $\SDA(\B(\D,D),\bar{\gamma})$ be defined as the largest value $d$
  such that there exists a finite set of distributions $\D_D \subseteq \D$ such that $\SDA(\D_D,D,\bar{\gamma}) = d$.
\end{definition}

Our technique gives the following lower bound for decision problems:
\begin{theorem}
\label{thm:avgvstat-random-decision}
  Let $D$ be a distribution and $\D$ be a set of distributions over a domain $X$ such that for some $\bar{\gamma}$, $\SDA(\B(\D,D),\bar{\gamma})=d$. Any (randomized) SQ algorithm that solves $\B(\D,D)$ with success probability $\alpha > 1/2$ requires at least $(2\alpha - 1) d$ queries to $\VSTAT(1/(3\bar{\gamma}))$.
\end{theorem}
\begin{proof}
 As before, we exhibit a hard distribution over input distributions for which every deterministic SQ algorithm that solves $\B(\D,D)$ with probability $\alpha$ (over the choice of input) requires at least $(2\alpha - 1) d$ queries to $\VSTAT(1/(3\bar{\gamma}))$.
 Let $\D_D$ be the set of distributions that witnesses the statistical dimension, namely,  $\SDA(\D_D,D,\bar{\gamma}) = d$. Consider the following distribution over the input distribution $D'$: $D'$ equals $D$ with probability $1/2$ and $D'$ equals a random uniform element of
  $\D_D$ with probability $1/2$.

$\A$ has success probability $\alpha > 1/2$ and therefore, when executed on $D$ with exact responses to queries, it must correctly identify $D$ (say it outputs 0 in this case). We then define $\D^+ \subseteq \D_D$ as the set of  distributions on which $\A$ is successful (that is outputs $1$). The probability of
success of $\A$ implies that $|\D^+| \geq (2\alpha -1) | \D_D|$. Now, the set $\D_D$ is included in the set of distributions $\D$ for which $0$ is not a valid solution. Therefore we can apply Lemma \ref{lem:hard-to-distinguish} with $\D_f = \D_D$ to obtain that the number of queries to  $\VSTAT(1/(3\bar{\gamma}))$ is $q \geq (2\alpha -1)d$.
\end{proof}

\subsection{Statistical Dimension Based on Pairwise Correlations}
\label{sec:pairwise-correlations}

In addition to $\SDA$ which is based on average correlation we introduce a simpler notion based on pairwise correlations. It is sufficient for some  applications and is easy to relate to
SQ-DIM used in learning (as we do in Section \ref{sec:app2SQ}). %For simplicity we only give a version of this notion for deterministic algorithms that does not require a bound on the solution set size.
\begin{definition} \label{def:correl}

We say that a set of $m$ distributions $\D = \{D_1,\ldots,D_m\}$ over $X$ is $(\gamma,\beta)$-correlated relative to a distribution $D$ over $X$ if:
  $$\left| \left\langle \frac{D_i}{D}-1, \frac{D_j}{D}-1
      \right\rangle_D \right| \leq
\begin{cases}
  &\beta  \mbox{ for } i=j\in[m]\\
  &\gamma \mbox{ for } i\not = j \in[m].
\end{cases}
$$
\end{definition}

\begin{definition}\label{def:sdim}
  For $\gamma,\beta>0$, domain $X$ and a search problem $\Z$ over a set of solutions $\F$
  and a class of distributions $\D$ over $X$, let $m$ be the largest integer such that there exists
  a {\em reference} distribution $D$ over $X$ and a finite set of distributions $\D_D \subseteq \D$ such that
  for any solution $f\in \F$, $\D_f=\D_D\setminus \Z_f$ is  $(\gamma,\beta)$-correlated relative to $D$ and $|\D_f| \geq m$.
We define the \textbf{statistical dimension} with pairwise correlations $(\gamma,\beta)$ of $\Z$
to be $m$ and denote it by $\SD(\Z,\gamma,\beta)$.

For decision problems $\SD(\B(\D,D),\gamma,\beta)$ is defined as the largest integer $m$
  such that there exists a set of distributions $\D_D \subseteq \D$ of size $m$ that is $(\gamma,\beta)$-correlated relative to $D$.
\end{definition}

It is easy to bound $\SDA$ of any $(\gamma,\beta)$-correlated set of distributions.
\begin{lemma}
\label{lem:cor-2-sda}
 Let $\D=  \{D_1,D_2,\ldots,D_m\}$ be a $(\gamma,\beta)$-correlated set of distributions relative to a distribution $D$. Then for every $\gamma' >0$,
  $\SDA(\D,D,\gamma+\gamma') \geq \frac{m \gamma'}{\beta - \gamma}$.
\end{lemma}
\begin{proof}
  Take $d = m \gamma'/(\beta - \gamma)$; we will prove that
  $\SDA(\D,D, \gamma +\gamma') \ge d$. Consider a set of distributions $\D' \subseteq \D$, where $|\D'| \ge |\D|/d \geq
  m/d = (\beta-\gamma) / \gamma'$:
  \begin{align*}
    \rho(\D',D) &= \frac{1}{|\D'|^2} \sum_{D_1,D_2 \in \D'}
    \left|\left\langle \frac{D_1}{D}-1, \frac{D_2}{D}-1
      \right\rangle_D \right| \\
    & \le \frac{1}{|\D'|^2} \left( |\D'| \beta + (|\D'|^2 - |\D'|) \gamma
    \right) \\
    & \le \gamma + \frac{\beta - \gamma}{|\D'|} \\
    & \le \gamma +\gamma'
  \end{align*}
\end{proof}

As an immediate corollary we obtain a bound on $\SDA$ of a search or decision problem from a bound on $\SD$.
\begin{corollary}
\label{cor:sd-genlb}
 Let $X$ be a domain and $\Z$ be a search or decision problem over a set of solutions $\F$
  and a class of distributions $\D$ over $X$. For $\gamma,\beta>0$, let $m = \SD(\Z,\gamma,\beta)$.  Then for every $\gamma' >0$,
  $\SDA(\Z, \gamma+ \gamma') \geq \frac{m \gamma'}{\beta - \gamma}$.
\end{corollary}

We now apply Theorem \ref{thm:avgvstat} to obtain the following lower bound on SQ algorithms in terms of $\SD$.
\begin{corollary}\label{cor:sd-genlb-decision}
 Let $X$ be a domain and $\Z$ be a search or decision problem over a set of solutions $\F$
  and a class of distributions $\D$ over $X$. For $\gamma,\beta>0$, let $m = \SD(\Z,\gamma,\beta)$.
  For any $\gamma' > 0$, any SQ algorithm requires at least $m \gamma'/(\beta-\gamma)$ queries to the $\STAT(\sqrt{\gamma+\gamma'})$ or $\VSTAT(1/(3(\gamma+\gamma')))$ oracle to solve $\Z$.
\end{corollary}
\eat{
Using Theorem \ref{thm:avgvstat-random-decision} we can also obtain a version of this bound for the decision problem $\B_{\mathrm{unif}}(\D,D)$.
\begin{corollary}
\label{cor:sd-random-decision}
  Let $D$ be a distribution and $\D$ be a set of $m$ distributions over a domain $X$ that is $(\gamma,\beta)$-correlated over $D$.  Any (randomized) SQ algorithm that solves $\B_{\mathrm{unif}}(\D,D)$ with success probability $\alpha > 1/2$ over the input distribution
  and randomness of the algorithm requires at least $(2\alpha - 1) m \gamma/(\beta-\gamma)$ calls to $\STAT(\sqrt{2\gamma})$ or $\VSTAT(1/(6\gamma))$ oracle.
\end{corollary}
}
In this corollary if, for example,  $\SD(\Z,\gamma=\frac{m^{-2/3}}{2},\beta =1) \geq m$ then at least  $m^{1/3}/2$ queries to $\VSTAT(m^{2/3}/3)$ or $\STAT(m^{-1/3})$ oracle are required to solve the problem.

\subsection{Lower Bounds for $1$-bit Sampling Algorithms}
\label{sec:vstat-to-sample}
Next we address lower bounds on algorithms that use the $\SAMPLE$ oracle. We recall that the $\SAMPLE$ oracle returns the value of a function on a single randomly chosen point. To
estimate the expectation of a function, an algorithm can simply query this oracle multiple times with the same function and average the results. %A lower bound for this oracle directly
%translates to a lower bound on the number of samples that any statistical algorithm must use.
%VF: not sure what the sentence above says formally so I'm commenting it out

We note that responses of $\SAMPLE$ do not have the room for the possibly adversarial deviation afforded by the tolerance of the $\STAT$ and $\VSTAT$ oracles. The ability to use
these slight deviations in a coordinated way is used crucially in our lower bounds against $\VSTAT$ and in all known lower bounds for SQ learning algorithms. While it is possible to derive lower bounds against \sampleoracle algorithms using $m$ queries from lower bounds against algorithms that use $O(m)$ queries to $\STAT(1/m)$ \citeP{Ben-DavidD98}, such lower bound will not suffice for our main application. It would only imply the trivial lower bound of $\Omega(n/k)$ queries to $\SAMPLE$ for the planted $k$-biclique problem.
Proving tighter lower bounds against \sampleoracle algorithms directly is harder and indeed lower bounds for the equivalent Honest SQ learning model required a substantially
more involved argument than lower bounds for the regular SQ model \citeP{Yang05}.

Our lower bounds for \sampleoracle algorithms rely on a direct simulation of the $\SAMPLE$ oracle using the $\VSTAT$ oracle. This simulation
%implies that the use of a \sampleoracle oracle does not give additional power to the algorithm and
allows us to derive lower bounds against \sampleoracle algorithms from Theorem
\ref{thm:avgvstat-random}. We also provide a reverse simulation of $\VSTAT$ oracle using $\SAMPLE$ oracle.

\eat{
 Next observe that $t \cdot p(1-p)$ is the variance of the
sum of $t$ random and independent samples of $h(x)$. For a sum of
Bernoulli random variables, this implies that an estimate of $\E_D[h]$
using $t$ samples will, with at least a constant probability, be
within $\sqrt{p(1-p)/t}$ of $\E_D[h]$. But also for $t > 1/(p(1-p))$,
with at least a constant probability, the error of the estimate will
be at least $\sqrt{p(1-p)/t}$, in other words one cannot get a better
estimate with high probability using $t$ samples. One extreme case we
also need to consider is when $t < 1/(p(1-p))$ (say $t<1/p$ for
simplicity). In this case an estimate obtained from $t$ samples is
likely to be 0 and we cannot distinguish (with more than a constant
probability) between two values $p$ and $p'$ which are both at most
$1/t$. Therefore it is necessary to set $\tau$ to be at least
$1/t$. It is also not hard to see that, for any two distributions $D$
and $D'$ and a function $h$, $\VSTAT(t)$ will give different answers
for query $h$ on distributions $D$ and $D'$ whenever it is possible to
distinguish (with at least a constant probability) $D$ from $D'$ using
the actual mean of $h$ on $t$ samples from $D$.
}

\begin{theorem}
\label{th:unbiased-from-vstat}
Let $\Z$ be a search problem and let $\A$ be a (possibly randomized) \sampleoracle algorithm that solves $\Z$ with probability at least $\alpha$ using $m$ samples from
$\SAMPLE$.  For any $\delta \in (0,1/4]$, there exists a SQ algorithm $\A'$ that uses at most $m$ queries to $\VSTAT(m/\delta^2)$ and solves $\Z$ with probability at least
$\alpha - \delta$.
\end{theorem}
Our proof relies on a simple simulation. Given query $h_1:X\rightarrow \{0,1\}$ from $\A$ to $\SAMPLE$, we make the same query $h_1$ to $\VSTAT(t)$ for $t=m/\delta^2$. Let
$p'_1$ be the response. We flip a coin with bias $p'_1$ (that is one that outputs $1$ with probability $p'_1$ and $0$ with probability $1-p'_1$) and return it to the algorithm. We do the
same for the remaining $m-1$ queries which we denote by $h_2,h_3,\ldots,h_m$. We then prove that the true $m$ samples of $\SAMPLE$ and our simulated coin flips are statistically
close by upper bounding the expected ratio of their density functions (which is equal to the $\chi^2$ divergence plus 1) . This implies that the success probability of the simulated algorithm is not much worse than that of the \sampleoracle algorithm.

In our proof we will, for simplicity and without loss of generality, assume that $\VSTAT(t)$ always outputs a value in the interval $[1/t,1-1/t]$. We can always replace a value $v$
returned by $\VSTAT(t)$ by $v'$ which is the closest value to $v$ in the above interval. It is easy to see that if $v$ is a valid answer of $\VSTAT(t)$ then so is $v'$.

We will need the following lemmas for our proof.
The first one bounds the total variation distance between two distributions in terms of the expected ratio of probability density functions.
\begin{lemma}
\label{lem:tv-from-ratio}
Let $D_1$ and $D_2$ be two distribution over a domain $X$ of finite\footnote{This assumption is simply for convenience of notation. It holds in our applications.} size such that
$D_2(x)$ is non-vanishing. Denote the total variation distance between $D_1$ and $D_2$ by $\Delta_{TV}(D_1,D_2)$.  Then $\Delta_{TV}(D_1,D_2) \leq \sqrt{\rho}/2$, where $\rho
= \E_{D_1}\left[\frac{D_1(x)}{D_2(x)}\right]-1$.
\end{lemma}
\begin{proof}
The key observation is that the $\chi^2$-divergence between $D_1$ and $D_2$ is exactly the expected ratio minus 1.
\alequn{\rho  = \E_{D_1}\left[\frac{D_1(x)}{D_2(x)}\right]-1 = \E_{D_2}\left[\frac{D_1^2(x)}{D_2^2(x)}\right]-1 = \E_{D_2}\left[\frac{D_1^2(x)}{D_2^2(x)} -
2\frac{D_1(x)}{D_2(x)}+1\right] = \E_{D_2}\left[\left(\frac{D_1(x)}{D_2(x)}-1\right)^2 \right] .}

By Jensen's inequality this implies that $$\E_{D_2}\left[\left|\frac{D_1(x)}{D_2(x)}-1\right| \right] \leq \sqrt{\E_{D_2}\left[\left(\frac{D_1(x)}{D_2(x)}-1\right)^2 \right]} =
\sqrt{\rho}\ .$$

Finally, $$\Delta_{TV}(D_1,D_2) = \frac{1}{2} \sum_{x\in X}|D_1(x) -D_2(x)| = \frac{1}{2} \E_{D_2}\left[\left|\frac{D_1(x)}{D_2(x)}-1\right| \right] \leq \frac{\sqrt{\rho}}{2}.
$$
\end{proof}

The second lemma is that if $p'$ is an answer of $\VSTAT(t)$ for a query  $h$, such that $\E_D[h] = p$, then the expected ratio of density functions of Bernoulli random variables with biases $p$ and $p'$, denoted
$B(p)$ and $B(p')$,  is small.
\begin{lemma}\label{lem:vstat-ratio}
For an integer $t$ and $p \in [0,1]$ let $p' \in [1/t,1-1/t]$ such that $|p'-p| \leq \max\left\{\frac{1}{t},\sqrt{\frac{p(1-p)}{t}}\right\}$. Then
\begin{align*}
\E_{b\sim B(p)}\left(\frac{\Pr[B(p) = b]}{\Pr[B(p')=b]}\right)
\leq 1 + \frac{3}{t} .
\end{align*}
\end{lemma}
\begin{proof}
If $b=1$, the ratio is $p/p'$ and when $b=0$, then it is $(1-p)/(1-p')$. Thus, the expected
ratio is $$\frac{p^2}{p'} + \frac{(1-p)^2}{1-p'} = 1+\frac{(p-p')^2}{p'(1-p')}\ .$$ We can assume without loss of generality that $p' \leq 1/2$.

Now if $p \leq 3 p'$ then $p(1-p) \leq 3 p'(1-p')$.  Otherwise (when, $p > 3p'$), we know that $p \geq 3p'\geq 3/t$. This implies that
  $p-p' \geq 2p/3 \geq 2/t$. This means that $$\frac{2p}{3} \leq p-p' \leq \sqrt{\frac{p(1-p)}{t}} \leq \sqrt{\frac{p}{t}}.$$
This can only be true when $p \leq (3/2)^2/t = 9/(4t)$, contradicting our assumption that $p\geq 3/t$.
 This implies that $$\max\left\{\frac{1}{t},\sqrt{\frac{p(1-p)}{t}}\right\} \leq \max\left\{\frac{1}{t},\sqrt{\frac{3p'(1-p')}{t}}\right\} \leq \sqrt{\frac{3p'(1-p')}{t}}\ .$$ By using
 this bound in the ratio equation we get that $$1+\frac{(p-p')^2}{p'(1-p')} \leq 1 + \frac{\frac{3p'(1-p')}{t}}{p'(1-p')} \leq 1+\frac{3}{t}\ .$$
\end{proof}

We can now complete the proof of Theorem \ref{th:unbiased-from-vstat}.
\begin{proof}[of Theorem \ref{th:unbiased-from-vstat}]
We simulate $\A$ using $\VSTAT(t)$ as described above. We now prove that for any algorithm the total variation distance between the true answers of $\SAMPLE$ and the simulated
distribution is at most $\delta$. Formally, let $R$ denote the set of all outcomes of $\A$'s random bits and for $r\in R$, let $\A^r$ denote the execution of $\A$ when its random bits
are set to $r$.
let $\Pi_\A$ denote the distribution over the $m$ bits obtained by the algorithm $\A$ when it is run with $\SAMPLE$ oracle. Similarly, let $\Pi'_\A$ denote the distribution over $\zo^m$
obtained by running the algorithm  $\A$ simulated using $\VSTAT(t)$ as above. By definition, $\Pi_\A = \E_{r\in R} \Pi_{\A^r}$ and similarly $\Pi'_\A = \E_{r\in R} \Pi'_{\A^r}$.
This implies that, $$\Delta_{TV}(\Pi_\A,\Pi'_\A) \leq \E_{r\in R}\left[\Delta_{TV}(\Pi_{\A^r},\Pi'_{\A^r})\right] \leq \max_{r\in R}\Delta_{TV}(\Pi_{\A^r},\Pi'_{\A^r}) .$$
The algorithm $\A^r$ is deterministic and it is therefore sufficient to prove the bound on total variation distance for deterministic algorithms. For conciseness we assume henceforth that
$\A$ is deterministic.

For any $i\in[m]$ let $\Pi_{\A_i}$ denote the probability distribution on the first $i$ samples of $\A$ executed with $\SAMPLE$.
For $j \leq i$ let $z^j$ denote the first $j$ bits of $z$. Let $\Pi_{\A_i}(z \cond z^{i-1})$ denote the probability that the first $i$ samples of $\A$ executed with $\SAMPLE$ oracle are
equal to $z$ conditioned on the probability that the first $i-1$ samples are equal to $z^{i-1}$. We define $\Pi'_{\A_i}(z)$  and $\Pi'_{\A_i}(z  \cond z^{i-1})$ analogously. We also
denote by $h_z$ the query that $\A$ asks after getting $z$ as the response to first $i$ samples and let $p_z = \E_D[h_z]$. Let $p'_z$ denote the response of $\VSTAT(t)$ on $h_z$.

For $i\in [m]$ and any $z \in \zo^i$, $\Pi_{\A_i}(z \cond z^{i-1}) = \Pr[B(p_{z^{i-1}})=z_i]$ and hence $\Pi_{\A_i}(z) = \Pi_{\A_{i-1}}(z^{i-1}) \Pr[B(p_{z^{i-1}})=z_i]$.
Similarly, $\Pi'_{\A_i}(z \cond z^{i-1}) = \Pr[B(p'_{z^{i-1}})=z_i]$ and $\Pi'_{\A_i}(z) = \Pi'_{\A_{i-1}}(z^{i-1}) \Pr[B(p'_{z^{i-1}})=z_i]$.
This implies that:
\alequn{\E_{z\sim \Pi_{\A_i}}\left[\frac{\Pi_{\A_i}(z)}{\Pi'_{\A_i}(z)}  \right] &= \E_{z\sim \Pi_{\A_i}}\left[\frac{\Pi_{\A_{i-1}}(z^{i-1})
\Pr[B(p_{z^{i-1}})=z_i]}{\Pi'_{\A_{i-1}}(z^{i-1}) \Pr[B(p'_{z^{i-1}})=z_i]}  \right] \\
&= \E_{y \sim \Pi_{\A_{i-1}}}\left[\frac{\Pi_{\A_{i-1}(y)}}{\Pi'_{\A_{i-1}(y)}}
\cdot \E_{b\sim B(p_y)}\left[\frac{\Pr[B(p_y)=b]}{\Pr[B(p'_y)=b]}\right] \right]\ .}
Now by Lemma \ref{lem:vstat-ratio}, this implies that for any $z$ of length  $i\in [m]$,
$$\E_{z\sim \Pi_{\A_i}}\left[\frac{\Pi_{\A_i}(z)}{\Pi'_{\A_i}(z)}  \right] \leq \E_{y \sim \Pi_{\A_{i-1}}}\left[\frac{\Pi_{\A_{i-1}(y)}}{\Pi'_{\A_{i-1}(y)}}  \right] \cdot
\left(1+\frac{2}{t}\right)\ .$$
Applying this iteratively we obtain that
$$\E_{z\sim \Pi_{\A}}\left[\frac{\Pi_\A(z)}{\Pi'_\A(z)}\right] \leq  \left(1+\frac{3}{t}\right)^m \leq e^{3m/t}\ .$$

By our definition, $t = m/\delta^2 \geq 16m$. Therefore, $3m/t \leq 1/5$ and hence $e^{3m/t} \leq 1+4m/t$.
By Lemma \ref{lem:tv-from-ratio}, we get that $\Delta_{TV}(\Pi_\A,\Pi'_\A) \leq \sqrt{(1+4m/t-1)}/2 = \sqrt{m/t} = \delta$. This implies that the success probability of $\A$ using
the simulated oracle is at least $\alpha - \delta$.

\end{proof}

We now combine Theorems \ref{thm:avgvstat-random-decision} and \ref{th:unbiased-from-vstat} to obtain the following lower bound for decision problems.
\begin{theorem}
\label{thm:avgsample-v-decision}
  Let $X$ be a domain and $D$ be a distribution over $X$ and $\D$ be a set of distributions over $X$. For $\bar{\gamma} > 0$, let $d = \SDA(\B(\D,D),\bar{\gamma})$. Any \sampleoracle algorithm that solves $\B(\D,D)$ with probability $\alpha$ requires at least $m$ queries to $\SAMPLE$ for $$m = \min\left\{\frac{d(2\alpha
  -1)}{2},\frac{(2\alpha+1)^2}{48\bar{\gamma}} \right\}\ .$$ In particular, any algorithm with success probability of at least $2/3$ requires at least
  $\min\{d/6,1/(432\bar{\gamma})\}$ queries to $\SAMPLE$.
\end{theorem}
\begin{proof}
Assuming the existence of a \sampleoracle algorithm using less than $m$ queries, we apply Theorem \ref{th:unbiased-from-vstat} for $\delta = (2\alpha-1)/4$ to simulate the
algorithm using $\VSTAT$. The bound on $m$ ensures that the resulting algorithm uses less than $d(2\alpha-1)/2$ queries to $\VSTAT(\frac{1}{3\bar{\gamma}})$ and has success probability of at least $\alpha - \delta = (2\alpha+1)/4$. By substituting these parameters into Theorem \ref{thm:avgvstat-random-decision} we obtain a contradiction.
\end{proof}

For general search problems this leads to the following lower bound.
\begin{theorem}
\label{thm:avgsample-v}
  Let $X$ be a domain and $\Z$ be a search problem over a set of solutions $\F$
  and a class of distributions $\D$ over $X$. For $\bar{\gamma} > 0$ and $\eta\in (0,1)$, let $d = \SDA(\Z,\bar{\gamma},\eta)$.
  Any (possibly randomized) \sampleoracle algorithm that solves $\Z$ with probability $\alpha$ requires at least $m$ calls to $\SAMPLE$ for $$m = \min\left\{\frac{d(\alpha
  -\eta)}{2(1-\eta)},\frac{(\alpha-\eta)^2}{12\bar{\gamma}} \right\}\ .$$ In particular, if $\eta \leq 1/6$ then any algorithm with success probability of at least $2/3$ requires at least
  $\min\{d/4,1/(48\bar{\gamma})\}$ queries to $\SAMPLE$.
\end{theorem}

To conclude, we formally state a simple reduction in the other direction, namely that $\VSTAT(t)$ oracle can be simulated using the $\SAMPLE$ oracle. It has been observed that, given a Boolean query function $h$ one can obtain an estimate of $\E_D[h]$ using $t = O(\log(1/\delta)/\tau^2)$ $1$-bit samples which with probability at
least $1-\delta$ will be within $\tau$ of $\E_D[h]$ \citeP{Ben-DavidD98}. Using the multiplicative Chernoff bound, it is not hard to see that $O(t\log(1/\delta))$ samples are sufficient to estimate $p=\E_D[h]$ within tolerance guaranteed by $\VSTAT(t)$. In addition, we will show how to use $\SAMPLE$ oracle to estimate the expectation of real-valued queries.

\begin{theorem}
\label{th:stat-from-unbiased}
Let  $t,q>0$ be any integers and $\delta>0$. There exists an algorithm $\A'$ that for any input distribution $D$ and any algorithm $\A$ that asks at most $q$ queries to $\VSTAT$, with probability at least $1-\delta$, provides valid for $\VSTAT(t)$ answers to all the queries of $\A$. $\A'$ uses $O(qt\cdot \log(q/\delta))$ queries to $\SAMPLE$ for the same input distribution $D$.
\end{theorem}
\begin{proof}
For every query $h:X \rightarrow [0,1]$ of $\A$, the algorithm $\A'$ estimates $p=\E_D[h]$ as follows. To generate a random Bernoulli variable with bias $p$, $\B$ draws $\theta \in [0,1]$ randomly and uniformly and defines: $h_\theta(x) = 1$ if $h(x) \leq \theta$ and $h_\theta(x)=0$ otherwise. It then makes the query $h_\theta$ to $\SAMPLE$. Observe that
$$\Pr_{\theta,x\sim D}[h_\theta(x)=1] = \E_{x\sim D}[\Pr_{\theta}[h_\theta(x)=1]]= \E_{x\sim D}[h(x)]=p.$$ The algorithm $\B$ repeats this $m$ times (each time choosing a new random $\theta$) and then answers the query $h$ with the mean of the obtained samples (for $m$ to be defined later). We denote the mean by $v$.

Assuming that $p \leq 1/2$, multiplicative Chernoff bounds imply that $$\Pr\left[|v-p| \geq \sqrt{p(1-p)/t}\right] \leq 2e^{-3mp/(p(1-p)t)}\leq 2e^{-6m/t} .$$
The bound in the case of $p>1/2$ follows from the symmetric argument.

Choosing $m=6t\cdot \ln(2q/\delta)$ ensures that $\Pr[|v-p| \geq \sqrt{p(1-p)/t}]  \leq \delta/q$. This implies that $q$ arbitrary queries for $\VSTAT(t)$ can be answered correctly with probability at least $1-\delta$ using $6qt\cdot \ln(2q/\delta)$ queries to $\SAMPLE$.
\end{proof}

\section{Warm-up: MAX-XOR-SAT}\label{sec:max-xor-sat}

In this section, we demonstrate our techniques on a warm-up problem, MAX-XOR-SAT. For this problem, it is sufficient to use pairwise correlations, rather than average correlations.

 For $\eps \geq 0$, the $\eps$-approximate \textbf{MAX-XOR-SAT} problem is defined as follows. Given samples from some unknown
 distribution $D$ over XOR clauses on $n$ variables, find an  assignment that maximizes up to additive error $\eps$ the probability a random clause drawn from $D$ is satisfied.

 In the worst case, it is known that MAX-XOR-SAT is NP-hard to
approximate to within $1/2-\delta$ for any constant $\delta$ \citeP{Hastad01}.  In practice, local search
algorithms such as WalkSat~\citeP{SelmanKC95} are commonly applied as
heuristics for maximum satisfiability problems.  We give strong evidence
that the distributional version of MAX-XOR-SAT is hard for
algorithms that locally seek to improve an assignment by flipping
variables as to satisfy more clauses, giving some theoretical
justification for the observations of~\citeP{SelmanKC95}.  Moreover,
our proof even applies to the case when there exists an assignment
that satisfies all the clauses generated by the target distribution.

The bound we obtain can be viewed as a restatement of the known lower bound for learning parities using statistical query algorithms (indeed, the problem of learning parities is a special
case of our distributional MAX-XOR-SAT).% \enote{is this true?}

 To formalize the search problem, we will denote by $C=\{0,1\}^n$ the set of XOR clauses in $n$ variables, such that for $c\in C$, if for $i\in [n]$ we have $c_i = 1$ then the $i$th variable appears in $c$, and otherwise it does not; for simplicity, no variables are negated in
the clauses.  Let $A=\{0,1\}^n$ denote the set of possible assignments to the variables. We will say that the assignment $a\in A$ satisfies the clause $c\in C$ if $a\cdot c=1$ (where $a\cdot c$ denotes the inner product modulo $2$).

Let ${\cal D}$ be the set of distributions over clauses in $C$. For a distribution $D\in {\cal D}$ and an assignment $a\in A$, let $f_D(a)=\E_{c\sim D} [a\cdot c]$ be the fraction of clauses that $a$ satisfies  under $D$. For $D\in {\cal D}$ let $M_D=\max_{a\in A} f_D(a)$.
The  MAX-XOR-SAT problem asks to find $a\in A$ that maximizes $f_D(a)$, given samples from an unknown distribution $D$.

We are now ready to formalize the search problem that we are interested in, using the  notation above and that of  Definition \ref{def:searchD}.

\begin{problem}\label{def:xorsat}($\eps$-approximate MAX-XOR-SAT)
Let $X=C=\{0,1\}^n$ (the set of clauses), $\D$ be the set of distributions over $X$,  $\F=A=\{0,1\}^n$ (the set of assignments). Let  $\Z: \D\rightarrow 2^\F$ be defined as $\Z(D)=\{a\in A \cond f_D(a)\geq M_D-\eps \}.$ %The goal is to find an $a\in {\cal Z}(D)$, using oracle access to a statistical oracle that can sample from the unknown distribution $D$.
\end{problem}

\begin{theorem}\label{thm:SDxorsat}
    For any $\delta>0$, any SQ algorithm requires at least $2^{n/3}-1$ queries to $\STAT(2^{-n/3})$ to solve $\left(\frac{1}{2}-\delta\right)$-approximate MAX-XOR-SAT.
    % In particular, at least $2^{n/3}$ queries of tolerance  $\tau={2^{-n/3}}$ are required.
\end{theorem}

We will first determine the statistical dimension of our search problem. This will immediately imply Theorem \ref{thm:SDxorsat} using Corollary \ref{cor:sd-genlb} (by choosing $\gamma' = 2^{-n/3}$).

\begin{lemma} For a $\delta>0$, let $\Z$ denote the $(1/2-\delta)$-approximate MAX-XOR-SAT. Then the statistical dimension of $\Z$ with pairwise correlation $(\gamma,\beta)=(0,1)$  is $\SD(\Z,0,1)\geq 2^n-1$.
\end{lemma}

\begin{proof}
We verify the properties of Definition \ref{def:sdim}.

Let the reference distribution $D=U_C$, the uniform distribution over $C=\{0,1\}^n$.
For $a\in A=\{0,1\}^n=\F$, let $D_a\in \D$ be the uniform distribution over $c\in C$ such that $a\cdot c=1$.  Let $\D_D=\{D_a\mid~ a\in A\}$, so $|\D_D|=2^n$.

For $a,b\in A$, we have
$$
f_{D_a}(b)=\E_{c\sim D_a}[b\cdot c]=\frac{1}{2^{n-1}}\left(  \sum_{c\in C ~\mid~a\cdot c=1, b\cdot c=1} 1\right).
$$
 Note that if $b=a$ then $f_{D_a}(b)=1$, and if $b\neq a$, $f_{D_a}(b)= 1/2$ (indeed, $|\{c\in C \mid a\cdot c=1, b\cdot c=1 \}| = 2^n/4$ since it is the size of two intersecting affine subspaces in $\{0,1\}^n$).

Therefore, for $a\in A$  and $\eps=1/2-\delta>0$, the set of solutions is $$\Z(D_a)=\{b\in A \cond f_{D_a}(b)\geq 1/2+\delta\}=\{a\}, $$ and so $\Z_a=\{D_a\}$.

To conclude the proof we will show that for any assignment $a\in \F=A$ the set
$\D_a=\D_D \setminus \{D_a\}$ of distributions is $(0,1)$-correlated (see Definition \ref{def:correl}).

Note that for $a\in A$, and $c\in C$,  $$\left(\frac{D_a}{D}-1\right)(c)=
\begin{cases}
  &-1\mbox{ if } c\cdot a=0\\
  &1 \mbox{ if } c\cdot a=1.
\end{cases}$$
In other words, $\left(\frac{D_a}{D}-1\right)(c) = -(-1)^{a \cdot c}$. A well-known (and easy to verify) property of $\{-1,1\}$-valued parity functions is that they are $(0,1)$-correlated over the uniform distribution. That is, for $a, b\in A$

 $$\left| \left\langle \frac{D_a}{D}-1, \frac{D_b}{D}-1
      \right\rangle_D \right| \leq
\begin{cases}
  &0  \mbox{ for } a=b\\
  &1 \mbox{ for } a\not = b .
\end{cases}
$$
\end{proof}

\section{Planted Biclique and Densest Subgraph}
\label{sec:strongclique}
\subsection{Statistical Dimension of Planted Biclique}
We now prove the lower bound claimed in
Theorem ~\ref{thm:planted-stat} on the
problem of detecting a planted $k$-biclique in the given distribution on vectors from $\{0,1\}^n$ as defined above.

Throughout this section we will use the following notation. For a subset $S \subseteq [n]$, let $D_{S}$ be the distribution over $\{0,1\}^n$ with a planted set $S$. Let $\cS_k$ denote the set of all ${n \choose k}$ subsets of
$[n]$ of size $k$ and $m = {n \choose k}$. We index the elements of $\cS_k$ in some arbitrary order as $S_1,\ldots,S_m$. For $i\in [m]$, we use $D_i$ to denote
$D_{S_i}$. We will also assume, whenever necessary, that $k$ and $n$ are larger than some fixed constant.

The reference distribution in our lower bounds will be the
uniform distribution over $\{0,1\}^n$ and let $\hat{D}_S$ denote
$D_S/D - 1$.  In order to apply our lower bounds based on statistical
dimension with average correlation we now prove that for the planted
biclique problem average correlations of large sets of distributions must be small. We
start with a lemma that bounds the correlation of two planted biclique
distributions relative to the reference distribution $D$ as a function
of the overlap between the planted sets:
\begin{lemma}
\label{lem:clique-correlation-bound}
For $i,j \in [m]$, $$\rho_D(D_i, D_j) = \left|\left\langle \hat{D}_i, \hat{D}_j \right\rangle_D\right| \leq \frac{2^\lambda k^2}{n^2}, $$ where $\lambda =  |S_i \cap S_j|$.
\end{lemma}

\begin{proof}
For the distribution $D_i$, we consider the probability $D_i(x)$ of
generating the vector $x$.  Then,
\begin{align*}
D_i(x) =
\begin{cases}
& (\frac{n-k}{n}) \frac{1}{2^n} + (\frac{k}{n})\frac{1}{2^{n-k}}  \mbox{\ \ \ \ \ \ \ \ \ \ \ \ \ \ \ \ \ \  if } \forall s \in S_i, x_s = 1\\
& (\frac{n-k}{n}) \frac{1}{2^n}  \mbox{\ \ \ \ otherwise} .
\end{cases}
\end{align*}
Now we compute the vector $\hat{D}_i = \frac{D_i}{D}-1$:
\begin{align*}
\frac{D_i}{D}-1 =
\begin{cases}
&  \frac{k 2^{k}}{n}  - \frac{k}{n}    \mbox{\ \ \ \ \ \ \ \ \ \  if } \forall s \in S_i, x_s = 1\\
&  -\frac{k}{n}\mbox{\ \ \ \ otherwise} .
\end{cases}
\end{align*}
We then bound the inner product:
\begin{eqnarray*}
\left\langle \hat{D}_i, \hat{D}_j \right\rangle_D & \leq &
\frac{2^{n-2k+\lambda}}{2^n}\left(\frac{k 2^{k}}{n}  -  \frac{k}{n} \right)^2
+ 2\left( \frac{2^{n-k}}{2^n} - \frac{2^{n-2k+\lambda}}{2^n}\right)\left( \frac{k 2^{k}}{n}  -
    \frac{k}{n} \right) \left(-\frac{k}{n}\right)
 + \left(-\frac{k}{n}\right)^2\\
 &\le& \frac{2^\lambda k^2}{n^2},
\end{eqnarray*}
which holds when $k \ge 3.$
We also note that $\left\langle \hat{D}_i, \hat{D}_j \right\rangle_D \geq 0$.
\end{proof}

We now give a bound on the average correlation of any $\hat{D}_S$ with a large number of distinct biclique distributions.
\begin{lemma}\label{lemma:pccalc}
 Let $\delta \geq 1/\log n$ and $k\leq n^{1/2 - \delta}$.  For any integer $\ell \leq k$,
   $S \in \cS_k$ and any set $A \subseteq \cS_k$
  where $|A| \ge 3(m-1)/n^{2\ell \delta}$,
\begin{align*}
 \frac{1}{|A|} \sum_{S_i\in A}\left| \la  \hat{D}_S, \hat{D}_i \ra \right| <
  2^{\ell+1}   \frac{k^2}{n^2}.
\end{align*}
\end{lemma}

\begin{proof}
In this proof we first show that if the total number of sets in $A$ is large then most of sets in $A$ have a small overlap with $S$. We then use the bound on the overlap of most sets to
obtain a bound on the average correlation of $D_S$ with distributions for sets in $A$.

Formally, we let $\alpha = \frac{k^2}{n^2}$ and using Lemma \ref{lem:clique-correlation-bound} get the bound $|\la \hat{D}_i,
\hat{D}_j \ra| \leq 2^{|S_i \cap S_j|} \alpha$. Summing over $S_i
\in A$,
\begin{align*}
  \sum_{S_i\in A} \left| \la  \hat{D}_S, \hat{D}_i \ra \right|\leq \sum_{S_i\in A} 2^{|S
    \cap S_i|} \alpha .
\end{align*}
For any set $A \subseteq \cS_k$ of size $t$ this bound is maximized when the sets of
$A$ include $S$, then all sets that intersect $S$ in $k-1$ indices,
then all sets that intersect $S$ in $k-2$ indices and so on until the
size bound $t$ is exhausted. We can therefore assume without loss of generality that $A$ is defined in precisely this way.

Let $T_\lambda = \left\{S_i \ | \ |S \cap S_i| =\lambda \right\}$ denote the subset of all $k$-subsets that intersect with $S$ in exactly $\lambda$ indices.
Let $\lambda_0$ be the smallest $\lambda$ for which $A \cap
T_\lambda$ is non-empty.  We first observe that for any $1 \leq j \leq k-1$,
\begin{align}
  \frac{|T_j|}{|T_{j+1}|} = \frac{{k \choose j} {n-k \choose k-j}}{{k
      \choose j+1} {n-k \choose k-j-1}} = \frac{(j+1) 
    (n-2k+j+1)}{(k-j)^2} \geq \frac{(j+1)(n-2k)}{(k-j)^2} \geq \nonumber\\
 \geq \frac{(j+1)(n-2n^{1/2-\delta})}{n^{1-2\delta}}\geq  
    \frac{(j+1)(1-2n^{-1/2-\delta})}{n^{-2\delta}}\geq
  \frac{(j+1)n^{2 \delta}}{2}.
  \label{eq:T-ratio}
\end{align}
By applying this equation inductively we obtain, $$|T_j| \leq
\frac{2^j \cdot |T_0|}{j! \cdot n^{2 \delta j}} < \frac{ 2^j \cdot
  (m-1)  }{j! \cdot n^{2 \delta j}}$$ where the last inequality holds since $|T_0| \leq m-2$ whenever $n\geq 2k+1$.
 For $n$ larger than some fixed constant
\begin{align*}
  \sum_{k \geq  \lambda \geq j} |T_{\lambda}| < \sum_{k \geq  \lambda \geq j} \frac{ 2^\lambda \cdot (m-1)  }{\lambda! \cdot n^{2\delta \lambda}} \leq \frac{m-1}{n^{2\delta
  j}} \sum_{k \geq  \lambda \geq j} \frac{ 2^\lambda  }{\lambda! \cdot n^{2\delta (\lambda-j)}} \leq \frac{3(m-1)}{n^{2\delta j}}\ .
\end{align*}
By definition of $\lambda_0$, $|A| \leq \sum_{j\geq \lambda_0} |T_j| <
3(m-1)/n^{2 \delta \lambda_0}$.  In particular, if $|A| \geq
3(m-1)/n^{2\ell \delta }$ then $n^{2\delta \lambda_0} < n^{2\ell\delta}$ or $\lambda_0 < \ell$.  Now we can
conclude that
\begin{align*}
  \sum_{S_i\in A} \left | \la \hat{D}_S, \hat{D}_i \ra \right |& \le \sum_{j
    =\lambda_0}^k 2^{j} |T_j
  \cap A| \alpha \\
  & \le \left( 2^{\lambda_0} |T_{\lambda_0} \cap A| +
    \sum_{j=\lambda_0+1}^k 2^j |T_j| \right)
  \alpha \\
  &\leq \left( 2^{\lambda_0} |T_{\lambda_0} \cap A| + 2 \cdot
    2^{\lambda_0+1} |T_{\lambda_0+1}| \right)
  \alpha \\
  & < 2^{\lambda_0+2} |A| \alpha \leq 2^{\ell+1} |A| \alpha.
\end{align*}
To derive the second to last inequality we need to note that for every $j \geq
0$, $2^j |T_j| > 2(2^{j+1}|T_{j+1}|)$ whenever $n^{2\delta} \geq 4$. We can therefore telescope the sum.
\end{proof}

We can now bound the statistical dimension (with average correlation) of the planted $k$-biclique problem.
\begin{theorem}
\label{th:clique-sda-bound}
For $\delta \geq 1/\log n$ and $k\leq n^{1/2 - \delta}$ let $\Z$ the distributional planted
$k$-biclique problem. Then for any $\ell \leq k$, $\SDA(\Z,
2^{\ell+1} k^2/n^2,1/{n \choose k}) \geq n^{2\ell\delta }/3$.
In addition, let $D$ be the uniform distribution and denote the set of all planted distributions
by $\D$. Then,  $\SDA(\D,D, 2^{\ell+1} k^2/n^2) \geq n^{2\ell\delta }/3$.
\end{theorem}

\begin{proof}
For every solution $S \in \F$, $\Z_S = \{D_S\}$ and let $\D_S = \D
  \setminus \{D_S\}$. Note that $|\D_S| = {n \choose k}-1$ and therefore $|\D_S| \geq (1-1/{n \choose k}) |\D|$. This means that we can use $1/{n \choose k}$ as the solution
  set bound.

   Let $\D'$ be a set of distributions $\D' \subseteq \D_{S}$ such that $|\D'| \geq 3(m-1)/n^{2\ell \delta}$. Then by Lemma \ref{lemma:pccalc}, for every $S_i \in \D'$,
  $$ \frac{1}{|\D'|} \sum_{S_j \in \D'} \left|\la \hat{D}_i, \hat{D}_j \ra\right| <
  2^{\ell+1} \frac{k^2}{n^2} .$$ In particular, $\rho(\D',D) <
  2^{\ell+1} \frac{k^2}{n^2}$. By the definition of $\SDA$ (Definition \ref{def:sdima-random}), this means
  that $\SDA(\Z, 2^{\ell+1} k^2/n^2,1/{n \choose k}) \geq n^{2 \ell \delta}/3$.

 The second claim holds by exactly the same argument since $|\D'| \geq m/d$ implies $|\D'| \geq (m-1)/d$.
\end{proof}

For a positive $r$ we choose $\ell = \log r-1$.  Our lower bound for the planted bi-clique problem stated in Theorem~\ref{thm:planted-stat} follows from substituting the bound $\SDA(\Z,
r k^2/n^2,1/{n \choose k}) \geq n^{2(\log(r)-1)\delta }/3$ into Theorem~\ref{thm:avgvstat-random} (with $\eta=1/{n \choose k}$ and $\delta = 2/3$). In addition,  by Theorem~\ref{thm:avgvstat-random-decision} used with $\alpha = 1/2 + 1/t$, we obtain hardness of the decision version of the problem for randomized SQ algorithms which also
implies Theorem~\ref{thm:planted-stat}.
\begin{theorem}\label{thm:planted-stat-decision}
For any constant $\delta > 0$, any $k \leq n^{1/2 - \delta}$ and $r > 0$, let $D$ be the uniform distribution over $\zo^n$ and $\D$ be the set of all planted $k$-biclique distributions. For some $t = n^{\Omega(\log{r})}$, any randomized SQ algorithm that solves the decision problem $\B(\D,D)$ with probability $1/2+1/t$ requires $t$ queries to $\VSTAT(n^2/(rk^2))$.
\end{theorem}

Theorem~\ref{th:unbiased-from-vstat} used with $\delta = 1/9$ implies that an  algorithm that uses $m$ queries to $\SAMPLE$ and has success probability $2/3$ gives an algorithm that  uses $m$ queries to $\VSTAT(81m)$ and has success probability $2/3-1/9 = 5/9$. For some $m = \Omega(n^2/k^2)$, Theorem \ref{thm:planted-stat-decision} applied with $r = \Omega(1)$ implies that such algorithm cannot exist. This implies the lower bound for \sampleoracle algorithms stated in Theorem~\ref{thm:planted-strong}.% by taking $\log r = 4/\delta$.

\subsection{Generalized Planted Densest Subgraph}
\label{sec:densest}

We will now show lower bounds on detecting a $(p,q)$-planted densest subgraph, a generalization of the
distributional planted biclique problem we defined in Definition \ref{def:densegraph}. Note that $p =1, q=1/2$ is precisely the distributional planted $k$-biclique problem. For this generalized problem, we will take $D$, the reference
distribution, to be that of $n$ independent Bernoulli variables with
bias $q$.

Before we give our results for this problem, we have to fix
some further notation: for $x \in \{0,1\}^n$, we define $\ones{x} =
\sum x_i$ ({\em i.e.}~the number of 1's in $x$); similarly for $\zeros{x} =
\sum 1-x_i$ (the number of 0's in $x$). We will denote the restriction
of a set by subscripting so that $x_S$ is $x$ restricted to the subset
$S \subseteq [n]$. We use $\bar{S}$ to denote the complement of $S$ in the current ground set.

First, we give a computation of the correlation. This is a generalized
version of Lemma \ref{lem:clique-correlation-bound}.
\begin{lemma}
Fix $0 < q \le p \le 1$ and let $\Delta_{pq} = 1+\frac{(p-q)^2}{q(1-q)}$. For $i,j \in [m]$,
\begin{align*}
  \innerprod{\hat{D}_i}{\hat{D}_{j}}_D = \left(
      \Delta_{pq}^{\lambda} -1 \right) \frac{
    k^2}{n^2},
\end{align*}
where $\lambda = |S_i \cap S_j|$.
\end{lemma}
\begin{proof}
 For any $x$, we
  have $D(x) = q^{\ones{x}} (1-q)^{\zeros{x}}$. For $D_i(x)$:
  \begin{align*}
    D_i(x) & = \Pr [x|\textrm{planted}]\Pr[\textrm{planted}] +
    \Pr[x|\textrm{not planted}]\Pr[\textrm{not planted}] \\
    & = \frac{k}{n} p^{\ones{x_{S_i}}}(1-p)^{\zeros{x_{S_i}}}
    q^{\ones{x_{\bar{S_i}}}} (1-q)^{\zeros{x_{\bar{S_i}}}} + \left( 1 -
      \frac{k}{n} \right) q^{\ones{x}} (1-q)^{\zeros{x}} .
  \end{align*}
  For $D_i(x)/D(x) - 1$, we have:
   \begin{align*}
     \frac{D_i(x)}{D(x)} -1 & = \frac{k}{n} \cdot
     \frac{p^{\ones{x_{S_i}}}(1-p)^{\zeros{x_{S_i}}} q^{\ones{x_{\bar{S_i}}}}
       (1-q)^{\zeros{x_{\bar{S_i}}}}}{q^{\ones{x}}
       (1-q)^{\zeros{x}}}-\frac{k}{n} \\
     & = \frac{k}{n} \left( \frac{p}{q}\right)^{\ones{x_{S_i}}} \left(
       \frac{1-p}{1-q} \right)^{\zeros{x_{S_i}}} - \frac{k}{n} .
   \end{align*}
  Now, for $S_j$ where $\abs{S_i \cap S_j} = \lambda$, we want to
  compute:
  \begin{align*}
    \innerprod{\hat{D}_i}{\hat{D}_{j}}_D = \left( \frac{k}{n}
    \right)^2 \sum_{ x \in \{ 0,1 \}^n} q^{\ones{x}}
    (1-q)^{\zeros{x}}\left[ \left( \frac{p}{q}\right)^{\ones{x_{S_i}}}
      \left( \frac{1-p}{1-q} \right)^{\zeros{x_{S_i}}} - 1 \right] \left[
      \left( \frac{p}{q}\right)^{\ones{x_{S_j}}} \left(
        \frac{1-p}{1-q} \right)^{\zeros{x_{S_j}}} - 1 \right] .
  \end{align*}
  There are three types of terms in the product in the summand. We
  deal with all these terms by repeated applications of the Binomial
  theorem. The first term illustrates this approach:
  \begin{align*}
    \sum_{ x \in \{0,1\}^n} q^{\ones{x}} (1-q)^{\zeros{x}} = (q +
    (1-q))^n = 1 .
  \end{align*}
  The second type of term is given by:
  \begin{align*}
    & \sum_{ x \in \{ 0,1 \}^n} q^{\ones{x}} (1-q)^{\zeros{x}} \left( \frac{p}{q}\right)^{\ones{x_{S_i}}}  \left(
    \frac{1-p}{1-q} \right)^{\zeros{x_{S_i}}}  \\
  & \qquad = \sum_{ x \in \{ 0,1 \}^n}
  q^{\ones{x_{\bar{S_i}}}}(1-q)^{\zeros{x_{\bar{S_i}}}} p^{\ones{x_{S_i}}}
  (1-p)^{\zeros{x_{S_i}}} \\
  & \qquad = \sum_{y \in \{ 0,1 \}^{\abs{S_i}}} p^{\ones{y}} (1-p)^{\zeros{y}}
  \sum_{z \in \{ 0,1 \}^{\abs{\bar{S_i}}}} p^{\ones{z}}
  (1-p)^{\zeros{z}} \\
  & \qquad = 1 .
  \end{align*}
  The third type of term is more complicated -- using the above trick,
  we can restrict $x$ to $T = S_i \cup S_j$ because the sum  taken over the
  remaining $x_i$ yields 1.
  \begin{align*}
\sum_{ x \in \{ 0,1 \}^n} q^{\ones{x}} (1-q)^{\zeros{x}}\left[ \left( \frac{p}{q}\right)^{\ones{x_{S_i}}}  \left(
    \frac{1-p}{1-q} \right)^{\zeros{x_{S_i}}}  \right]
\left[ \left( \frac{p}{q}\right)^{\ones{x_{S_j}}}  \left(
    \frac{1-p}{1-q} \right)^{\zeros{x_{S_j}}} \right]  \\
 = \sum_{ x \in \{ 0,1 \}^{\abs{T}}} q^{\ones{x}} (1-q)^{\zeros{x}}\left[ \left( \frac{p}{q}\right)^{\ones{x_{S_i}}}  \left(
    \frac{1-p}{1-q} \right)^{\zeros{x_{S_i}}}  \right]
\left[ \left( \frac{p}{q}\right)^{\ones{x_{S_j}}}  \left(
    \frac{1-p}{1-q} \right)^{\zeros{x_{S_j}}} \right] .
\end{align*}
Similarly, we can sum $x$ over coordinates in $S_i\setminus S_j$ and
$S_j \setminus S_i$. Hence, the sum simplifies:
\begin{align*}
  & \sum_{ x \in \{ 0,1 \}^n} q^{\ones{x}} (1-q)^{\zeros{x}}\left[
    \left( \frac{p}{q}\right)^{\ones{x_{S_i}}} \left( \frac{1-p}{1-q}
    \right)^{\zeros{x_{S_i}}} \right] \left[ \left(
      \frac{p}{q}\right)^{\ones{x_{S_j}}} \left(
      \frac{1-p}{1-q} \right)^{\zeros{x_{S_j}}} \right]  \\
  & = \sum_{ x \in \{ 0,1 \}^{\abs{S_i \cap S_j}}} q^{\ones{x}}
  (1-q)^{\zeros{x}}\left[ \left( \frac{p}{q}\right)^{\ones{x}}
    \left(
      \frac{1-p}{1-q} \right)^{\zeros{x}}  \right]^2 \\
  & = \sum_{ x \in \{ 0,1 \}^{\abs{S_i \cap S_j}}} \left(
    \frac{p^2}{q}\right)^{\ones{x}} \left(
    \frac{(1-p)^2}{1-q} \right)^{\zeros{x}} \\
  & = \left( \frac{p^2}{q} + \frac{(1-p)^2}{1-q} \right)^{\lambda} \\
  & = \Delta_{pq}^\lambda .
\end{align*}
Combining these three calculations yields:
\begin{align*}
  \innerprod{\hat{D}_i}{\hat{D}_{j}}_D =  \left( \frac{k}{n} \right)^2
  \left( \Delta_{pq}^\lambda
    - 1\right)
\end{align*}
\end{proof}
Next, in analogy with Lemma \ref{lemma:pccalc}, we give a bound on
average correlation for sufficiently many distributions.

\begin{lemma}
  Fix $0 < q < p \le 1$ and let $\Delta_{pq} = 1+\frac{(p-q)^2}{q(1-q)}$.  For $\delta > 0$ and $k\leq n^{1/2- \delta}$, if $n^{2\delta} \ge 8 \Delta_{pq}$ then for any integer $\ell \leq k$,   $S\in \cS_k$ and $A \subseteq \cS_k$ of size at least $2(m-1)/n^{2\ell \delta}$,
\begin{align*}
  \frac{1}{|A|} \sum_{S_i\in A} \left|\innerprod{ \hat{D}_S}{ \hat{D}_i}\right| <
   \frac{2k^2}{n^2} \left(  \Delta_{pq}^\ell -1 \right) .
\end{align*}
\end{lemma}
\begin{proof}
  We proceed as in the proof of Lemma \ref{lemma:pccalc}. Recall that $T_\lambda = \left\{S_i \ | \ |S \cap S_i| = \lambda \right\}$ denotes the subset of all $k$-subsets that intersect with $S$ in exactly $\lambda$ indices.
Let $\lambda_0$ be the smallest $\lambda$ for which $A \cap
T_\lambda$ is non-empty. As before, we obtain that $\lambda_0 < \ell$.

We now bound the average correlation with $\hat{D}_S$ as follows:
 \begin{align*}
    \sum_{S_i\in A} \left|\la \hat{D}_S, \hat{D}_i \ra \right| & \le
    \sum_{j=\lambda_0}^k \frac{k^2}{n^2} \cdot \left( \Delta_{pq}^j - 1 \right)
    |T_j \cap A| \\
    & \le \frac{k^2}{n^2} \cdot \left( |T_{\lambda_0} \cap A| \left(\Delta_{pq}^{\lambda_0} - 1 \right) +
      \sum_{j=\lambda_0+1}^k |T_j| (\Delta_{pq}^j - 1) \right) .
  \end{align*}

To bound the sum
  \begin{align*}
    \sum_{j=\lambda_0+1}^k (\Delta_{pq}^j-1) \abs{T_j}
  \end{align*}
  it suffices to show that it is geometrically decreasing as:
  \begin{align*}
    (\Delta_{pq}^j -1)\abs{T_j} \ge 2 \cdot (\Delta_{pq}^{j+1}-1) \abs{T_{j+1}} .
  \end{align*}
 We first note that $\Delta_{pq} > 1$ and therefore for $j \geq 1$,
 $$\frac{\Delta_{pq}^{j+1}-1}{\Delta_{pq}^j-1} \leq \Delta_{pq}+1 < 2\Delta_{pq}.$$

  From equation \eqref{eq:T-ratio} in the proof of Lemma \ref{lemma:pccalc} and our assumption on $\Delta_{pq}$ we obtain the necessary property:
  \begin{align*}
    \frac{\abs{T_j}}{\abs{T_{j+1}}} \ge \frac{(j+1)n^{2\delta}}{2}
    \ge 4 \cdot \Delta_{pq} > \frac{2(\Delta_{pq}^{j+1}-1)}{\Delta_{pq}^j-1}.
  \end{align*}
  To conclude,

  \begin{align*}
    \sum_{S_i\in A} \left|\la \hat{D}_S, \hat{D}_i \ra \right| & \le \frac{k^2}{n^2} \cdot \left( |T_{\lambda_0} \cap A| \left(\Delta_{pq}^{\lambda_0} - 1 \right) +
      \sum_{j=\lambda_0+1}^k |T_j| (\Delta_{pq}^j -1) \right) \\
    &\leq\frac{k^2}{n^2} \cdot \left( |T_{\lambda_0} \cap A| \left(\Delta_{pq}^{\lambda_0} - 1 \right) + 2 \cdot
      |T_{\lambda_{0}+1}| (\Delta_{pq}^{\lambda_0+1}- 1)\right)
     \\
    & \le 2 \cdot \frac{k^2}{n^2} \cdot \abs{A} \left(\Delta_{pq}^{\lambda_0+1} -1 \right)
         \\
    & \le 2 \cdot \frac{k^2}{n^2} \cdot \abs{A} \left(\Delta_{pq}^\ell -1 \right).
  \end{align*}
\end{proof}

From here the bound on statistical dimension $\SDA$ of detecting the $(p,q)$-planted densest subgraph now follows in the same way as in Theorem \ref{th:clique-sda-bound}.
\begin{theorem}
  Fix $0 < q < p \le 1$.  For $\delta > 0$ and $k\leq
  n^{1/2-\delta}$ let $\Z$ be the distributional $(p,q)$-planted
  densest $k$-subgraph problem. Then for any $\ell \leq
  k$, $$\SDA\left(\Z, \frac{2k^2}{n^2}  \left( \Delta_{pq}^\ell -1 \right),\frac{1}{{n \choose k}}\right)
  \geq n^{2 \ell \delta }/2.$$ provided that $ n^{2\delta} \ge 8\Delta_{pq}$.
\end{theorem}

This $\SDA$ bound yields lower bounds for the $\VSTAT$ oracle:
\begin{corollary}\label{cor:densevstat}
  Fix $0 < q < p \le 1$  and let $\Delta_{pq} = 1+\frac{(p-q)^2}{q(1-q)}$.  For any constant $\delta > 0$, any $k \le
  n^{1/2 - \delta}$, $\ell \leq k$, at least $n^{\Omega(\ell)}$
  queries to $\VSTAT(n^2/(6k^2(\Delta_{pq}^\ell-1)))$ are required
  to solve the distributional $(p,q)$-planted
  densest $k$-subgraph problem with probability at least $2/3$ provided that $ n^{2\delta} \ge 8\Delta_{pq}$.
\end{corollary}
Similarly, by Theorem \ref{thm:avgvstat-random-decision}, the same lower bound applies to the decision version of the problem.

% This leads to interesting bounds for specific choices of $p$ and
% $q$. For example, the following gives a lower bound for various sparse
% cases by taking $l = O(\log n)$:
% \begin{corollary}
%   Fix $0 < q < p \le 1$ where $ q = c/n^t$ and $p=d/n^t$.  For $ \delta
%   > 0$, and $k \leq n^{1/2 - \delta}$, any unbiased statistical
%   algorithm requires $\tilde{\Omega}((cn^{2+t})/k^2 (d-c)^2)$ samples to
%   find a generalized planted densest subgraph of size $k$.
% \end{corollary}
% In particular, we have a lower bound of $\tilde{\Omega}(n^3 / k^2 )$
% for the sparse case when $p,q = O(1/n)$. This lower bound generalizes:
% in case $p$ and $q$ are slightly larger, or when they are not of the
% same order. Thus, we can state a corollary for general $p$ and $q$:
% \begin{corollary}
%   Fix $0 < q \le p \le 1$.  For $ \delta > 0$, and $k \leq n^{1/2 -
%     \delta}$, any unbiased statistical algorithm requires
%   $\tilde{\Omega}(q(1-q)n^2)/ l k^2 (p-q)^2)$ samples to find a
%   generalized planted densest subgraph of size $k$ provided that $
%   n^{\delta} \ge 1 + (p-q)^2/q(1-q)$.
% \end{corollary}

One is often interested in the case when $q=1/2$ and $p=1/2+\alpha$ (the
classical planted densest $k$-subgraph problem). In this setting $\Delta_{pq} = 1+4\alpha^2$ and $\Delta_{pq}^\ell-1 \leq e^{4\alpha^2\ell}- 1 \leq 8 \alpha^2\ell$ whenever $\ell \leq 1/(4 \alpha^2)$.
This gives a lower bound of $n^{\Omega(\ell)}$ against $\VSTAT(n^2/(48\ell \alpha^2 k^2))$ as stated in Corollary \ref{thm:densest-stat-intro}.

Finally, we give an example of a corollary for the $\SAMPLE$ oracle.
\begin{corollary}\label{cor:denseunbias}
  For constants $c, \delta > 0$, density $p = {1}/{2} + {1}/{n^c}$, and $k \leq n^{1/2 - \delta}$,
  Let $D$ be the uniform distribution over $\zo^n$ and $\D$ be the set of all $(p,1/2)$-planted densest $k$-subgraph distributions. Any (randomized) \sampleoracle algorithm that solves the decision problem $\B(\D,D)$ with probability at least $2/3$, requires
  $\Omega((n^{2+2c})/{k^2})$
  queries to $\SAMPLE$.
\end{corollary}
\begin{proof}
  By the argument above with $\alpha = 1/n^c$, $\SDA(\B(\D,D), 16 k^2 \ell/n^{2+2c}) \geq n^{2\ell\delta}/2$. For $\ell = 4/2\delta$ we obtain that $\SDA(\B(\D,D), 64 k^2/n^{2+2c}) \geq n^4/2$. By applying Theorem \ref{thm:avgsample-v-decision} for success probability $2/3$, we obtain a lower bound of 
  $$m = \min\left\{\frac{d/3}{2},\frac{(4/3+1)^2}{48\bar{\gamma}} \right\} = \min\left\{\frac{n^4}{12},\frac{49/9}{48} \cdot \frac{n^{2+2c}}{64 k^2} \right\} = \Omega\left( \frac{n^{2+2c}}{k^2}\right)\ .$$ 
  samples to $\SAMPLE$.
\end{proof}

\section{Applications to Statistical Query Learning}\label{sec:app2SQ}
We will now use Corollary~\ref{cor:sd-genlb-decision} to demonstrate that our results generalize the
notion of statistical query dimension in learning theory and the statistical query lower
bounds based on $\SQDIM$.  We then show that our lower bounds imply stronger and more general lower bounds in the context of learning.

We start with a few relevant definitions. In an instance of a PAC learning
problem, the learner has access to random examples of an unknown
boolean function $c: X'\rightarrow \{-1,1\}$ from a set of Boolean
functions $\C$. A random example is a pair including a point and its
label $(x',c(x'))$ such that $x'$ is drawn randomly from a
distribution $D'$, which might or might not be known to the learning algorithm (whenever necessary, we use $'$ to distinguish
variables from the identically named ones in the context of
general search problems). Specifically, for a target function $c \in C$ and distribution $D'$ over $X'$ we denote by $D_c$ over $X=X' \times \{-1,1\}$, where $D_c(x',c(x')) = D'(x')$ and $D_c(x',-c(x')) =0$.

For $\eps > 0$, the goal of an $\eps$-accurate learning algorithm is to find,
with high probability, a Boolean hypothesis $h$ for which $\Pr_{x'
  \sim D'}[h(x') \neq c(x')] \leq \eps$. A statistical query learning algorithm~\citeP{Kearns:98} has access to the $\STAT$ oracle for the input distribution $D_c$ in place of random examples.

%VF I think this is redundant
\eat{
A statistical query learning algorithm~\citeP{Kearns:98} has access to a statistical query oracle for the unknown function $c$ and distribution $D'$ in place of random examples. A query to the SQ oracle is a function $\phi: X' \times \{-1,1\} \rightarrow [-1,1]$. To such a query the oracle returns a value $v$ which is within $\tau$ of $\E_{D'}[\phi(x',c(x')]$, where $\tau$ is the tolerance parameter. This oracle is exactly the $\STAT$ oracle for the input distribution $D_c$.
}

\subsection{Relationship to $\SQDIM$}
\label{sec:relate-2-sq}

\citeAN{BlumFJKMR94} defined the {\em statistical query
  dimension} or \SQDIM of a set of functions $\C$ and distribution
$D'$ over $X'$ as follows (we present a simplification and
strengthening due to \citeAN{Yang05}).
\begin{definition}[\citeP{BlumFJKMR94}]
\label{def:old-weak-sqd}
For a concept class $\C$ and distribution $D'$, $\SQDIM(\C,D') = d'$ if $d'$ is the largest value for which there exist $d'$ functions $c_1, c_2, \ldots, c_{d'} \in \C$ such that for every
$i\neq j$, $|\la c_i , c_j \ra_{D'}| \leq 1/d'$.
\end{definition}

We first observe that correlations of Boolean functions relative to a distribution $D'$ are equivalent to correlations of corresponding distributions over examples relative to some reference distribution. Namely, let the reference distribution $D$ be the distribution for
which for every $(x',\ell) \in X$, $D(x',\ell) = D'(x')/2$. This is the distribution in which points are distributed according to $D'$ and labels are random unbiased coin flips. We denote it by $D' \times \{1/2,1/2\}$.
\begin{lemma}
For a distribution $D'$ and any Boolean functions $c,c_1$ and $c_2$,
For all $x' \in X'$, $\frac{D_c(x',\ell)}{D(x',\ell)}-1 = \ell \cdot c(x')$ and
$$\left\la \frac{D_{c_1}}{D} - 1, \frac{D_{c_2}}{D} - 1 \right\ra_D =  \la c_1, c_2 \ra_{D'}.$$
\end{lemma}
\begin{proof}
We first note that the definition of $D$ ensures that $D(x',\ell)$ is non-vanishing only when $D'(x')$ is non-vanishing and hence the function $\left(
\frac{D_c}{D} - 1 \right)$ is well-defined for any Boolean $c\in \C$. For every $c\in \C$, we have
$$\frac{D_c(x',c(x'))}{D(x',c(x'))}-1 = 2-1=1 \mathrm{\ \ \ and\ \ \ } \frac{D_c(x',-c(x'))}{D(x',-c(x'))}-1 = 0-1= -1.$$

Therefore,
$\frac{D_c(x',\ell)}{D(x',\ell)} -1= \ell \cdot c(x')$.
This implies that for any $c_1,c_2 \in \C$,
$$\left\la \frac{D_{c_1}}{D} - 1, \frac{D_{c_2}}{D} - 1 \right\ra_D = \E_{(x',\ell) \sim D}[\ell \cdot c_1(x') \cdot \ell \cdot c_2(x')] = \E_{D'}[c_1(x') \cdot
c_2(x')] = \la c_1, c_2 \ra_{D'}.$$
\end{proof}

The direct implication of this is that if $\SQDIM(\C,D') = d'$ then there exist $d'$ distributions over examples that are $(1/d',1)$-correlated relative to $D$. In particular, the decision problem of distinguishing example distributions from $D$ has large statistical dimension with pairwise correlations. We state this formally:
\begin{theorem}
\label{thm:sqdim2sd}
For a concept class $\C$ and distribution $D'$ over $X$ let $d' = \SQDIM(\C,D')$. Then for $\D_\C = \{D_c \cond c\in \C\}$ and $D=D' \times \{1/2,1/2\}$, $\SD(\B(\D_\C,D),1/d',1) \geq d'$.
\end{theorem}

\citeAN{BlumFJKMR94} proved that if a class of functions is
learnable using only a polynomial number of statistical queries of
inverse polynomial tolerance then its statistical query dimension is
polynomial. \citeAN{Yang05} strengthened their result and proved
the following bound (see \citeP{Szorenyi09} for a simpler proof).
\begin{theorem}[\citeP{Yang05}]
\label{th:SQDIM}
Let $\C$ be a class of functions and $D'$ be a distribution over $X'$ and let $d'=\SQDIM(\C,D')$.
Any SQ algorithm that learns $\C$ over $D'$ with error $\eps < 1/2-1/(2d'^{1/3})$ requires at least $d'^{1/3}/2-1$
queries to $\STAT(1/d'^{1/3})$.
\end{theorem}

In this result, the distribution $D'$ is fixed and known to the learner (such learning is referred to as {\em distribution-specific}) and it can be used to lower bound the complexity of learning $\C$ even in a weak sense. Specifically, when the learning algorithm is only required to output a hypothesis $h'$ such that $\Pr_{x' \sim D'}[h'(x') \neq c(x')] \leq 1/2 - \gamma'$ for some inverse polynomial $\gamma'$. It is well-known that weak learning of functions from $\C$ implies ability to distinguish examples of any function in $\C$ from points labeled randomly. This implies that we can apply our lower bound for decision problems to obtain a lower bound for weak learning that is essentially the same as the result of \citeAN{Yang05}.
\begin{corollary}
\label{cor:sqdim-from-sd}
Let $\C$ be a class of functions and $D'$ be a distribution over $X'$ and let $d'=\SQDIM(\C,D')$.
Any SQ algorithm that learns $\C$ over $D'$ with error $\eps < 1/2-1/d'^{1/3}$ requires at least $2d'^{1/3}-1$
queries to $\STAT(1/d'^{1/3})$.
\end{corollary}
\begin{proof}
Let $\D_\C = \{D_c \cond c\in \C\}$ and $D=D' \times \{1/2,1/2\}$.
We convert the weak learning algorithm into the algorithm for $\B(\D_\C,D)$ as follows.
Run the weak learning algorithm. Given hypothesis $h$ estimate the prediction error within $d'^{-1/3}/2$ by using the query $\phi(x',\ell) = h(x') \cdot \ell$ with tolerance $d'^{-1/3}$. If the answer to the query is $> d'^{-1/3}$ output 1 (meaning the input distribution is in $\D_\C$), otherwise output 0 (meaning that the input distribution is $D$).
Note that $\E_D[\phi] = 0$ and therefore this algorithm will always output 0 on $D$. Further, if the input distribution is $D_c$ and $\Pr_{D'}[h(x') \neq c(x')] < 1/2-1/d'^{1/3}$ then $$\E_{(x',\ell)\sim D_c}[\phi(x',\ell)] = \E_{(x',\ell)\sim D_c}[h(x') \cdot \ell] = \E_{x' \sim D'}[h(x') \cdot c(x')] = 1-2\Pr_{x' \sim D'}[h(x') \neq c(x')] > 2/d'^{1/3} .$$
Therefore in this case the answer to the query will be $> 2/d'^{1/3} - 1/d'^{1/3} = 1/d'^{1/3}$ and the algorithm will output 1.

By Theorem \ref{thm:sqdim2sd}, $\SD(\B(\D_\C,D),1/d',1) \geq d'$.
We can now apply the lower bound in Corollary \ref{cor:sd-genlb-decision} with $\gamma' = d'^{-2/3}/2$ to obtain that our algorithm must use $2d'^{1/3}$ queries to $\STAT(d'^{-1/3})$ to solve the problem. Our algorithm used one more query than the learning algorithm (of the same tolerance) which gives the stated lower bound.
\end{proof}
This corollary implies that lower bounds based on $\SQDIM$ are a special case of our lower bounds. One can also similarly show that the lower bounds based on the statistical query dimension of \citeAN{Feldman:12jcss} that characterizes learning to high accuracy are also a special case of our lower bounds.

\subsection{Lower Bounds for $1$-bit Sampling Oracle}
We now show how our results can be used to obtain lower bounds against \sampleoracle algorithms based on $\SQDIM$. Such lower bounds have been previously proved by \citeAN{Yang05} who referred to his model as Honest SQ model (apparently unaware of the connection to the model in \citeP{Ben-DavidD98}). In the Honest SQ model, the learner has access to an $\HSQ$ oracle.
%and can again evaluate queries which are a function
% of the data points and their labels.  As in $\SAMPLE$ oracle,
%the queries are evaluated on an fresh sample drawn from the target distribution.  More precisely,
A query to $\HSQ$ oracle is a function $\phi: X' \times \{-1,1\} \rightarrow \{-1,1\}$ and a sample
size $t > 0$. The oracle draws $x'_1, \ldots, x'_t \sim D'$, and returns
the value $\frac{1}{t}\sum_{i=1}^{t}{\phi(x',c(x'))}$. The total sample complexity of an algorithm
is the sum of the sample sizes it passes to $\HSQ$.

We note that using $\SAMPLE$ is equivalent to a call to $\HSQ$ with sample size $1$. Also $\SAMPLE$ can simulate estimation of queries from a larger number of samples in a straightforward way while obtaining the same
total sample complexity. Therefore $\HSQ$ is equivalent to the $\SAMPLE$ oracle.

Using Lemma \ref{lem:cor-2-sda} with $\gamma' = 1/\sqrt{d'}$ to convert a bound on pairwise correlations to a bound on average correlation, we can obtain that  $\SDA(\B(\D_\C,D),1/\sqrt{d'} + 1/d',1) \geq \sqrt{d'}/(1-1/d')$. Plugging this bound into Theorem ~\ref{thm:avgsample-v-decision}, we can derive sample complexity bounds on
\sampleoracle algorithms for learning in the same way as in the proof of Corollary \ref{cor:sqdim-from-sd}.
\begin{corollary}
Let $\C$ be a class of functions, $D'$ be a distribution over $X'$, $d'=\SQDIM(\C,D')$ and $\eps =1/2-1/d'^{1/4}$. Then
any \sampleoracle algorithm that, with probability at least $2/3$, $\eps$-accurately learns $\C$ over $D'$ requires $\Omega(\sqrt{d'})$ queries to $\SAMPLE$.
\end{corollary}
This lower bound is similar to the result of \citeAN{Yang05} who shows a bound of $\Omega(d'/\log{d'})$ using a stronger $1/d'^3$ upper bound on correlations (and a substantially more involved proof). Note that the inverse of the maximum pairwise correlation is usually much lower than the number of functions. Therefore our result will give a stronger lower bound in most cases.

\subsection{New Lower Bound for Learning}
We now briefly describe a version of our lower bound for weak distribution-specific learning. It is stronger than known $\SQDIM$-based bounds in several ways. First, it explicitly decouples the tolerance (or number of samples) from the number of queries. This is particularly relevant for {\em attribute-efficient} learning that is learning when the dimension is high but the target function depends on few variables (see \citeP{Feldman14:open} for more details on SQ learning in this setting). Second, it captures sample complexity in a tighter way by going to average correlations and proving lower bounds against $\VSTAT$. Lower bounds against $\VSTAT$ also imply tighter lower bounds for $\SAMPLE$ and, via the reductions in \citeP{FeldmanPV:13}, against stronger oracles.
%VF This is too minor so commented out.
%Finally, our lower bounds hold against randomized algorithms which was not known for the bounds based on $\SQDIM$.

We now give versions of our main definitions specialized to the case of distribution-specific PAC learning. Although the target distribution is fixed, by varying the concept by which examples are labeled, we effectively generate a large set of different distributions as before.
The average correlation can be defined directly for a set of functions $\C'$ relative to a distribution $D'$:
\begin{align*}
\rho(\C',D') \doteq
\frac{1}{|\C'|^2}\sum_{c_1,c_2 \in \C'} \left|\left\langle
    c_1, c_2 \right\rangle_{D'} \right| .
\end{align*}

\begin{definition}\label{def:sdima-learning}
  For $\bar{\gamma}>0$, a distribution $D'$ over domain $X'$ and a set of Boolean functions $\C$ over $X'$
  the \textbf{statistical dimension} of $\C$ over $D'$ with average correlation $\bar{\gamma}$ is defined to be
   the largest integer $d$ for which there exists a finite set of functions $\C_{\bar{\gamma}} \subseteq \C$ such that
   for any subset $\C' \subseteq \C_{\bar{\gamma}}$, where $|\C'| \ge \C_{\bar{\gamma}}/d$, $\rho(\C',D') \leq \bar{\gamma}$. We denote it by $\SDA(\C,D',\bar{\gamma})$.
\end{definition}

Using Theorem \ref{thm:avgvstat-random-decision} and the reduction in Corollary \ref{cor:sqdim-from-sd} imply Theorem \ref{thm:general-learning}.

\section*{Acknowledgments} We thank Benny Applebaum, Avrim Blum, Uri Feige, Ravi Kannan, Michael Kearns, Robi Krauthgamer, Moni Naor, Jan Vondrak, and Avi Wigderson for
insightful comments and helpful discussions.

\bibliography{stat_algs}

\begin{thebibliography}{71}
\providecommand{\natexlab}[1]{#1}
\providecommand{\url}[1]{\texttt{#1}}
\expandafter\ifx\csname urlstyle\endcsname\relax
  \providecommand{\doi}[1]{doi: #1}\else
  \providecommand{\doi}{doi: \begingroup \urlstyle{rm}\Url}\fi

\bibitem[Alon et~al.(2007)Alon, Andoni, Kaufman, Matulef, Rubinfeld, and
  Xie]{AlonAKMRX07}
N.~Alon, A.~Andoni, T.~Kaufman, K.~Matulef, R.~Rubinfeld, and N.~Xie.
\newblock Testing k-wise and almost k-wise independence.
\newblock In \emph{STOC}, pages 496--505, 2007.

\bibitem[Alon et~al.(1998)Alon, Krivelevich, and Sudakov]{AlonKS98}
Noga Alon, Michael Krivelevich, and Benny Sudakov.
\newblock Finding a large hidden clique in a random graph.
\newblock In \emph{SODA}, pages 594--598, 1998.

\bibitem[Ames and Vavasis(2011)]{AmesV11}
Brendan P.~W. Ames and Stephen~A. Vavasis.
\newblock Nuclear norm minimization for the planted clique and biclique
  problems.
\newblock \emph{Math. Program.}, 129\penalty0 (1):\penalty0 69--89, 2011.

\bibitem[Applebaum et~al.(2010)Applebaum, Barak, and Wigderson]{ApplebaumBW10}
Benny Applebaum, Boaz Barak, and Avi Wigderson.
\newblock Public-key cryptography from different assumptions.
\newblock In \emph{STOC}, pages 171--180, 2010.

\bibitem[Arora et~al.(2010)Arora, Barak, Brunnermeier, and Ge]{AroraBBG10}
Sanjeev Arora, Boaz Barak, Markus Brunnermeier, and Rong Ge.
\newblock Computational complexity and information asymmetry in financial
  products (extended abstract).
\newblock In \emph{ICS}, pages 49--65, 2010.

\bibitem[Bartlett and Mendelson(2002)]{BartlettMendelson:02}
P.~Bartlett and S.~Mendelson.
\newblock {Rademacher and Gaussian} complexities: Risk bounds and structural
  results.
\newblock \emph{Journal of Machine Learning Research}, 3:\penalty0 463--482,
  2002.

\bibitem[Belloni et~al.(2009)Belloni, Freund, and Vempala]{BelloniFV09}
Alexandre Belloni, Robert~M. Freund, and Santosh Vempala.
\newblock An efficient rescaled perceptron algorithm for conic systems.
\newblock \emph{Math. Oper. Res.}, 34\penalty0 (3):\penalty0 621--641, 2009.

\bibitem[Ben{-}David and Dichterman(1998)]{Ben-DavidD98}
Shai Ben{-}David and Eli Dichterman.
\newblock Learning with restricted focus of attention.
\newblock \emph{J. Comput. Syst. Sci.}, 56\penalty0 (3):\penalty0 277--298,
  1998.

\bibitem[Berthet and Rigollet(2013)]{BerthetR:13}
Quentin Berthet and Philippe Rigollet.
\newblock Complexity theoretic lower bounds for sparse principal component
  detection.
\newblock In \emph{COLT}, pages 1046--1066, 2013.

\bibitem[Bhaskara et~al.(2010)Bhaskara, Charikar, Chlamtac, Feige, and
  Vijayaraghavan]{BhaskaraCCFV10}
Aditya Bhaskara, Moses Charikar, Eden Chlamtac, Uriel Feige, and Aravindan
  Vijayaraghavan.
\newblock Detecting high log-densities: an {\it o}({\it n}$^{\mbox{1/4}}$)
  approximation for densest {\it k}-subgraph.
\newblock In \emph{STOC}, pages 201--210, 2010.

\bibitem[Bhaskara et~al.(2012)Bhaskara, Charikar, Vijayaraghavan, Guruswami,
  and Zhou]{BhaskaraCVGZ12}
Aditya Bhaskara, Moses Charikar, Aravindan Vijayaraghavan, Venkatesan
  Guruswami, and Yuan Zhou.
\newblock Polynomial integrality gaps for strong sdp relaxations of densest
  {\it k}-subgraph.
\newblock In \emph{SODA}, pages 388--405, 2012.

\bibitem[Blum et~al.(2005)Blum, Dwork, McSherry, and Nissim]{BlumDMN:05}
A.~Blum, C.~Dwork, F.~McSherry, and K.~Nissim.
\newblock {Practical privacy: the SuLQ framework}.
\newblock In \emph{PODS}, pages 128--138, 2005.

\bibitem[Blum et~al.(1994)Blum, Furst, Jackson, Kearns, Mansour, and
  Rudich]{BlumFJKMR94}
Avrim Blum, Merrick~L. Furst, Jeffrey~C. Jackson, Michael~J. Kearns, Yishay
  Mansour, and Steven Rudich.
\newblock Weakly learning dnf and characterizing statistical query learning
  using fourier analysis.
\newblock In \emph{STOC}, pages 253--262, 1994.

\bibitem[Blum et~al.(1998)Blum, Frieze, Kannan, and Vempala]{BlumFKV98}
Avrim Blum, Alan~M. Frieze, Ravi Kannan, and Santosh Vempala.
\newblock A polynomial-time algorithm for learning noisy linear threshold
  functions.
\newblock \emph{Algorithmica}, 22\penalty0 (1/2):\penalty0 35--52, 1998.

\bibitem[Bresler et~al.(2014)Bresler, Gamarnik, and Shah]{BreslerGS14a}
Guy Bresler, David Gamarnik, and Devavrat Shah.
\newblock Structure learning of antiferromagnetic ising models.
\newblock In \emph{NIPS}, pages 2852--2860, 2014.

\bibitem[Brubaker and Vempala(2009)]{BV}
S.~Brubaker and S.~Vempala.
\newblock Random tensors and planted cliques.
\newblock In \emph{Approximation, Randomization, and Combinatorial
  Optimization. Algorithms and Techniques}, volume 5687, pages 406--419. 2009.

\bibitem[{Cai} et~al.(2015){Cai}, {Liang}, and {Rakhlin}]{CaiLR15}
T.~T. {Cai}, T.~{Liang}, and A.~{Rakhlin}.
\newblock {Computational and Statistical Boundaries for Submatrix Localization
  in a Large Noisy Matrix}.
\newblock \emph{ArXiv e-prints}, February 2015.

\bibitem[Chu et~al.(2006)Chu, Kim, Lin, Yu, Bradski, Ng, and
  Olukotun]{ChuKLYBNO:06}
C.~Chu, S.~Kim, Y.~Lin, Y.~Yu, G.~Bradski, A.~Ng, and K.~Olukotun.
\newblock Map-reduce for machine learning on multicore.
\newblock In \emph{NIPS}, pages 281--288, 2006.

\bibitem[Coja-Oghlan(2010)]{Coja10}
Amin Coja-Oghlan.
\newblock Graph partitioning via adaptive spectral techniques.
\newblock \emph{Combinatorics, Probability {\&} Computing}, 19\penalty0
  (2):\penalty0 227--284, 2010.

\bibitem[Dekel et~al.(2011)Dekel, Gurel-Gurevich, and Peres]{DekelGP:11}
Y.~Dekel, O.~Gurel-Gurevich, and Y.~Peres.
\newblock Finding hidden cliques in linear time with high probability.
\newblock In \emph{ANALCO}, pages 67--75, 2011.

\bibitem[Dempster et~al.(1977)Dempster, Laird, and Rubin]{DempsterLR77}
A.~P. Dempster, N.~M. Laird, and D.~B. Rubin.
\newblock Maximum likelihood from incomplete data via the em algorithm.
\newblock \emph{Journal of the Royal Statistical Society, Series B},
  39\penalty0 (1):\penalty0 1--38, 1977.

\bibitem[Deshpande and Montanari(2013)]{DeshpandeMontanari13}
Yash Deshpande and Andrea Montanari.
\newblock Finding hidden cliques of size {\textbackslash}sqrt\{N/e\} in nearly
  linear time.
\newblock \emph{CoRR}, abs/1304.7047, 2013.

\bibitem[Deshpande and Montanari(2015)]{DeshpandeM15}
Yash Deshpande and Andrea Montanari.
\newblock Improved sum-of-squares lower bounds for hidden clique and hidden
  submatrix problems.
\newblock In \emph{{COLT}}, pages 523--562, 2015.
\newblock URL \url{http://jmlr.org/proceedings/papers/v40/Deshpande15.html}.

\bibitem[Dughmi(2014)]{Dughmi14}
Shaddin Dughmi.
\newblock On the hardness of signaling.
\newblock In \emph{{FOCS}}, pages 354--363, 2014.

\bibitem[Dunagan and Vempala(2008)]{DunaganV08}
John Dunagan and Santosh Vempala.
\newblock A simple polynomial-time rescaling algorithm for solving linear
  programs.
\newblock \emph{Math. Program.}, 114\penalty0 (1):\penalty0 101--114, 2008.

\bibitem[Feige and Krauthgamer(2000)]{FeigeK00}
U.~Feige and R.~Krauthgamer.
\newblock Finding and certifying a large hidden clique in a semirandom graph.
\newblock \emph{Random Struct. Algorithms}, 16\penalty0 (2):\penalty0 195--208,
  2000.

\bibitem[Feige and Ron(2010)]{FeigeRon:10}
U.~Feige and D.~Ron.
\newblock Finding hidden cliques in linear time.
\newblock In \emph{AofA}, pages 189--204, 2010.

\bibitem[Feige(2002)]{Feige02}
Uriel Feige.
\newblock Relations between average case complexity and approximation
  complexity.
\newblock In \emph{IEEE Conference on Computational Complexity}, page~5, 2002.

\bibitem[Feige and Krauthgamer(2003)]{feige2003probable}
Uriel Feige and Robert Krauthgamer.
\newblock The probable value of the {Lov{\'a}sz--Schrijver} relaxations for
  maximum independent set.
\newblock \emph{SICOMP}, 32\penalty0 (2):\penalty0 345--370, 2003.

\bibitem[Feldman(2008)]{Feldman:08ev}
V.~Feldman.
\newblock Evolvability from learning algorithms.
\newblock In \emph{STOC}, pages 619--628, 2008.

\bibitem[Feldman(2012)]{Feldman:12jcss}
V.~Feldman.
\newblock A complete characterization of statistical query learning with
  applications to evolvability.
\newblock \emph{Journal of Computer System Sciences}, 78\penalty0 (5):\penalty0
  1444--1459, 2012.

\bibitem[Feldman(2014)]{Feldman14:open}
Vitaly Feldman.
\newblock Open problem: The statistical query complexity of learning sparse
  halfspaces.
\newblock In \emph{COLT}, pages 1283--1289, 2014.

\bibitem[Feldman(2016)]{Feldman:16sqd}
Vitaly Feldman.
\newblock A general characterization of the statistical query complexity.
\newblock \emph{CoRR}, abs/1608.02198, 2016.
\newblock URL \url{http://arxiv.org/abs/1608.02198}.

\bibitem[Feldman et~al.(2013)Feldman, Perkins, and Vempala]{FeldmanPV:13}
Vitaly Feldman, Will Perkins, and Santosh Vempala.
\newblock On the complexity of random satisfiability problems with planted
  solutions.
\newblock \emph{CoRR}, abs/1311.4821, 2013.
\newblock Extended abstract in STOC 2015.

\bibitem[Feldman et~al.(2015)Feldman, Guzman, and Vempala]{FeldmanGV:15}
Vitaly Feldman, Cristobal Guzman, and Santosh Vempala.
\newblock Statistical query algorithms for stochastic convex optimization.
\newblock \emph{CoRR}, abs/1512.09170, 2015.
\newblock URL \url{http://arxiv.org/abs/1512.09170}.

\bibitem[Frieze and Kannan(2008)]{FriezeK08}
Alan~M. Frieze and Ravi Kannan.
\newblock A new approach to the planted clique problem.
\newblock In \emph{FSTTCS}, pages 187--198, 2008.

\bibitem[{Gao} et~al.(2014){Gao}, {Ma}, and {Zhou}]{GaoMZ14}
C.~{Gao}, Z.~{Ma}, and H.~H. {Zhou}.
\newblock {Sparse CCA: Adaptive Estimation and Computational Barriers}.
\newblock \emph{ArXiv e-prints}, September 2014.

\bibitem[Gelfand and Smith(1990)]{GelfandSmith90}
A.~E. Gelfand and A.~F.~M. Smith.
\newblock Sampling based approaches to calculating marginal densities.
\newblock \emph{Journal of the American Statistical Association}, 85:\penalty0
  398--409, 1990.

\bibitem[Hajek et~al.(2015)Hajek, Wu, and Xu]{HajekWX15}
Bruce~E. Hajek, Yihong Wu, and Jiaming Xu.
\newblock Computational lower bounds for community detection on random graphs.
\newblock In \emph{{COLT}}, pages 899--928, 2015.
\newblock URL \url{http://jmlr.org/proceedings/papers/v40/Hajek15.html}.

\bibitem[H{\aa}stad(2001)]{Hastad01}
Johan H{\aa}stad.
\newblock Some optimal inapproximability results.
\newblock \emph{J. ACM}, 48:\penalty0 798--859, July 2001.
\newblock ISSN 0004-5411.

\bibitem[Hastings(1970)]{Hastings70}
W.~K. Hastings.
\newblock Monte carlo sampling methods using markov chains and their
  applications.
\newblock \emph{Biometrika}, 57\penalty0 (1):\penalty0 97--109, 1970.

\bibitem[Hazan and Krauthgamer(2011)]{HazanK11}
Elad Hazan and Robert Krauthgamer.
\newblock How hard is it to approximate the best nash equilibrium?
\newblock \emph{SIAM J. Comput.}, 40\penalty0 (1):\penalty0 79--91, 2011.

\bibitem[Jerrum(1992)]{Jerrum92}
Mark Jerrum.
\newblock Large cliques elude the metropolis process.
\newblock \emph{Random Struct. Algorithms}, 3\penalty0 (4):\penalty0 347--360,
  1992.

\bibitem[Juels and Peinado(2000)]{JuelsP00}
Ari Juels and Marcus Peinado.
\newblock Hiding cliques for cryptographic security.
\newblock \emph{Des. Codes Cryptography}, 20\penalty0 (3):\penalty0 269--280,
  2000.

\bibitem[Kannan()]{Kannan-personal}
Ravi Kannan.
\newblock personal communication.

\bibitem[Karp(1979)]{Karp:79}
R.~Karp.
\newblock Probabilistic analysis of graph-theoretic algorithms.
\newblock In \emph{Proceedings of Computer Science and Statistics 12th Annual
  Symposium on the Interface}, page 173, 1979.

\bibitem[Kasiviswanathan et~al.(2011)Kasiviswanathan, Lee, Nissim,
  Raskhodnikova, and Smith]{KasiviswanathanLNRS11}
Shiva~Prasad Kasiviswanathan, Homin~K. Lee, Kobbi Nissim, Sofya Raskhodnikova,
  and Adam Smith.
\newblock What can we learn privately?
\newblock \emph{SIAM J. Comput.}, 40\penalty0 (3):\penalty0 793--826, June
  2011.

\bibitem[Kearns(1998)]{Kearns:98}
M.~Kearns.
\newblock Efficient noise-tolerant learning from statistical queries.
\newblock \emph{Journal of the ACM}, 45\penalty0 (6):\penalty0 983--1006, 1998.

\bibitem[Khot(2004)]{Khot04a}
Subhash Khot.
\newblock Ruling out ptas for graph min-bisection, densest subgraph and
  bipartite clique.
\newblock In \emph{FOCS}, pages 136--145, 2004.

\bibitem[Kirkpatrick et~al.(1983)Kirkpatrick, Jr., and Vecchi]{KirkpatrickGV83}
Scott Kirkpatrick, D.~Gelatt Jr., and Mario~P. Vecchi.
\newblock Optimization by simmulated annealing.
\newblock \emph{Science}, 220\penalty0 (4598):\penalty0 671--680, 1983.

\bibitem[Kucera(1995)]{Kucera95}
Ludek Kucera.
\newblock Expected complexity of graph partitioning problems.
\newblock \emph{Discrete Applied Mathematics}, 57\penalty0 (2-3):\penalty0
  193--212, 1995.

\bibitem[Ma and Wu(2013)]{MaW13a}
Zongming Ma and Yihong Wu.
\newblock Computational barriers in minimax submatrix detection.
\newblock \emph{CoRR}, abs/1309.5914, 2013.
\newblock URL \url{http://arxiv.org/abs/1309.5914}.

\bibitem[McSherry(2001)]{McSherry01}
F.~McSherry.
\newblock Spectral partitioning of random graphs.
\newblock In \emph{FOCS}, pages 529--537, 2001.

\bibitem[Meka et~al.(2015)Meka, Potechin, and Wigderson]{MekaPW15}
R.~Meka, A.~Potechin, and A.~Wigderson.
\newblock Sum-of-squares lower bounds for planted clique.
\newblock In \emph{{STOC}}, pages 87--96, 2015.

\bibitem[Metropolis et~al.(1953)Metropolis, Rosenbluth, Rosenbluth, Teller, and
  Teller]{MetropolisRRTT53}
Nicholas Metropolis, Arianna~W. Rosenbluth, Marshall~N. Rosenbluth, Augusta~H.
  Teller, and Edward Teller.
\newblock Equations of state calculations by fast computing machines.
\newblock \emph{Journal of Chemical Physics}, 21:\penalty0 1087--1092, 1953.

\bibitem[Minder and Vilenchik(2009)]{MV:09}
L.~Minder and D.~Vilenchik.
\newblock Small clique detection and approximate nash equilibria.
\newblock 5687:\penalty0 673--685, 2009.

\bibitem[Pearson(1900)]{Pearson:1900}
K.~Pearson.
\newblock On the criterion that a given system of deviations from the probable
  in the case of a correlated system of variables is such that it can be
  reasonably supposed to have arisen from random sampling.
\newblock \emph{Philosophical Magazine, Series 5}, 50\penalty0 (302):\penalty0
  157--175, 1900.

\bibitem[Selman et~al.(1995)Selman, Kautz, and Cohen]{SelmanKC95}
Bart Selman, Henry Kautz, and Bram Cohen.
\newblock Local search strategies for satisfiability testing.
\newblock In \emph{DIMACS Series in Discrete Mathematics and Theoretical
  Computer Science}, pages 521--532, 1995.

\bibitem[Servedio(2000)]{Servedio:00jcss}
R.~Servedio.
\newblock Computational sample complexity and attribute-efficient learning.
\newblock \emph{Journal of Computer and System Sciences}, 60\penalty0
  (1):\penalty0 161--178, 2000.

\bibitem[Steinhardt et~al.(2016)Steinhardt, Valiant, and Wager]{SteinhardtVW16}
J.~Steinhardt, G.~Valiant, and S.~Wager.
\newblock Memory, communication, and statistical queries.
\newblock In \emph{{COLT}}, pages 1490--1516, 2016.

\bibitem[Steinhardt and Duchi(2015)]{SteinhardtD15}
Jacob Steinhardt and John~C. Duchi.
\newblock Minimax rates for memory-bounded sparse linear regression.
\newblock In \emph{COLT}, pages 1564--1587, 2015.
\newblock URL \url{http://jmlr.org/proceedings/papers/v40/Steinhardt15.html}.

\bibitem[Sz{\"o}r{\'e}nyi(2009)]{Szorenyi09}
Bal{\'a}zs Sz{\"o}r{\'e}nyi.
\newblock Characterizing statistical query learning: Simplified notions and
  proofs.
\newblock In \emph{ALT}, pages 186--200, 2009.

\bibitem[Tanner and Wong(1987)]{TannerW87}
M~Tanner and W~Wong.
\newblock The calculation of posterior distributions by data augmentation (with
  discussion).
\newblock \emph{Journal of the American Statistical Association}, 82:\penalty0
  528--550, 1987.

\bibitem[Valiant(1984)]{Valiant84}
Leslie~G. Valiant.
\newblock A theory of the learnable.
\newblock \emph{Commun. ACM}, 27\penalty0 (11):\penalty0 1134--1142, 1984.

\bibitem[Vapnik and Chervonenkis(1971)]{VapnikChervonenkis:71}
V.~Vapnik and A.~Chervonenkis.
\newblock On the uniform convergence of relative frequencies of events to their
  probabilities.
\newblock \emph{Theory of Probab. and its Applications}, 16\penalty0
  (2):\penalty0 264--280, 1971.

\bibitem[\v{C}ern\'{y}(1985)]{Cerny1985Thermodynamical}
V.~\v{C}ern\'{y}.
\newblock {Thermodynamical approach to the traveling salesman problem: An
  efficient simulation algorithm}.
\newblock \emph{Journal of Optimization Theory and Applications}, 45\penalty0
  (1):\penalty0 41--51, January 1985.
\newblock ISSN 0022-3239.

\bibitem[{Wang} et~al.(2014){Wang}, {Berthet}, and {Samworth}]{WangBS15}
T.~{Wang}, Q.~{Berthet}, and R.~J. {Samworth}.
\newblock Statistical and computational trade-offs in estimation of sparse
  principal components.
\newblock \emph{ArXiv e-prints}, August 2014.

\bibitem[Yang(2001)]{Yang:01}
Ke~Yang.
\newblock On learning correlated boolean functions using statistical queries.
\newblock In \emph{ALT}, pages 59--76, 2001.

\bibitem[Yang(2005)]{Yang05}
Ke~Yang.
\newblock New lower bounds for statistical query learning.
\newblock \emph{J. Comput. Syst. Sci.}, 70\penalty0 (4):\penalty0 485--509,
  2005.

\bibitem[Yao(1977)]{Yao:1977}
Andrew Yao.
\newblock Probabilistic computations: Toward a unified measure of complexity.
\newblock In \emph{FOCS}, pages 222--227, 1977.

\bibitem[Zhang et~al.(2013)Zhang, Duchi, Jordan, and Wainwright]{ZhangDJW13}
Yuchen Zhang, John~C. Duchi, Michael~I. Jordan, and Martin~J. Wainwright.
\newblock Information-theoretic lower bounds for distributed statistical
  estimation with communication constraints.
\newblock In \emph{{NIPS}}, pages 2328--2336, 2013.
\newblock URL
  \url{http://papers.nips.cc/paper/4902-information-theoretic-lower-bounds-for-distributed-statistical-estimation-with-communication-constraints}.

\end{thebibliography}

\appendix

\section{Average-case vs Distributional Planted Bipartite Clique}
\label{sec:avg-to-dist-clique}

In this section we show the equivalence between the average-case
planted biclique problem (where a single graph is chosen randomly) and the
distributional biclique problem (where a bipartite graph is obtained from independent samples over $\{0,1\}^n$). The primary issue is that in the distributional biclique problem the biclique does not necessarily have the same size on the left side of
vertices as it does on the right side. We show that this is easy to fix by producing planted bicliques of smaller size on one of the sides. We do this by replacing vertices of the graph with randomly connected ones. We now describe the reductions more formally.

\begin{definition}\label{apbc}[Average-case planted biclique $\apbc(n,k_1,k_2)$]
Given integers $1\leq k_1,k_2\leq n$, consider the following distribution $\cald_{avg}(n,k_1,k_2)$  on bipartite graphs on $[n]\times[n]$ vertices. Pick two random sets of $k_1$ and $k_2$ vertices each from left and right side, respectively, say $S_1$ and $S_2$. Plant a bipartite clique on $S_1\times S_2$ and add an edge between all other pairs of vertices with probability $1/2$. The problem is to recover $S_1$ and $S_2$ given a random graph sampled from $\cald_{avg}(n,k_1,k_2)$.
\end{definition}

We will refer to the distributional biclique problem with $n$ samples as $\dpbc(n,k)$. Recall that in this problem we are given $n$ random and independent samples from distribution $D_S$  over $\zo^n$ for some unknown $S\subset [n]$ of size $k$ (see Definition \ref{def:clique}). The goal is to recover $S$.

\begin{theorem}\label{thm:distributional2average}
Suppose that there is an algorithm that solves $\apbc(n, k',k')$ in time $T'(n,k')$ and outputs the correct answer with probability $p'(n,k')$. Then there exists an algorithm that solves $\dpbc(n,k)$
in time $T(n,k)=O(nk T'(n,k/2))$ and outputs the correct answer with probability $p(n,k)=p'(n,k/2) - n2^{-\Omega(k)}$.
\end{theorem}
\begin{proof}
We will think of the distribution $\cald_{avg}(n,k',k')$ on graphs as a distribution on their respective adjacency matrices from $\{0,1\}^{n\times n}$.
Let $\A(n,k')$ be the algorithm that solves an instance of $\apbc(n,k',k')$.  Given $k$ and $n$, and access to
$n$ samples from $D_S$ for some set $S$ of size $k$, we will design an algorithm that finds $S$ by making  $O(nk)$ calls to the algorithm $\A(n,k')$ that
solves an instance of $\apbc(n,k',k')$.

Let $M$ be the $n\times n$ binary matrix whose rows are the $n$ samples from $D_S$.
First apply a random permutation $\pi:[n]\rightarrow [n]$ to the columns of $M$ to obtain $M'$ (this will ensure that the planted set is uniformly distributed among the $n$ coordinates, which is necessary in order to obtain instances distributed according to $\cald_{avg}(n,k',k')$).

In what follows we will denote by $k'\times k$ a biclique with $k'$ vertices on the left and $k$ vertices on the right.
Note that $M'$ has a $k'\times k$ planted biclique for some $k'$ that is distributed according to the binomial distribution $B(n,k/n)$. We denote the vertices on the left side of this biclique by $L$.
By a multiplicative Chernoff bound, $\Pr[ k/2\leq k'\leq 2 k ]\geq 1-2e^{-k/8}$.  From now on we will condition on this event occurring.

We first suppose that $k\leq k' \leq 2k$. We aim at obtaining instances of $\apbc(n,k,k)$ but recall that the left side of the planted bliclique has size $k' \geq k$. To reduce the size of the left side of the planted biclique to $k$ we will be replacing the vertices on the left side by randomly connected ones, one-by-one in a random order.
That is, start with $M'_0 = M'$. To obtain $M'_{t+1}$, we choose a random and uniform row of $M'_t$ that was not previously picked and replace it with a random and uniform $\{0,1\}^n$ vector. This gives a sequence of random matrices: $M'_1, M'_2, \ldots, M'_n$. Clearly, $M'_0$ has a planted biclique of size $k' \times k$ and $M'_n$ does not have a planted biclique (or, equivalently, has a $0 \times k$ biclique). A single step reduces the size of the left side of the biclique by at most 1. Therefore for some $i^*$, $M'_{i^*}$ has a $k\times k$ biclique. We denote the left side of this biclique by $L^*$. It is also easy to see that for every step $i$, conditioned on $M'_i$ containing $k''$ out of vertices in $L$, $M'_i$ is distributed exactly according to $\cald_{avg}(n,k'',k)$. We now condition on the event that for all $i$, $M'_i$ does not contain a $k\times k$ biclique such that its right side is different from $\pi(S)$.  It is not hard to see that this event happens with probability at least $1-n 2^{-\Omega(k)}$.

To recover $S$, we run $\A(n,k)$ on all the matrices $M'_i$. Let $L_i \times S_i$ be the biclique that $\A$ outputs on $M'_i$. We verify that $L_i \times S_i$ is a $k\times k$ biclique in $M'_i$. If so we output $\pi^{-1}(S_i)$. Note that when executed on $M'_{i^*}$, with probability $p'(n,k)$ this procedure will return $L^* \times \pi(S)$. In this case we will return exactly $S$. Further, by our conditioning, if the output of $\A$ is a $k\times k$ biclique then its right side must be $\pi(S)$.

We can now assume that $k/2 \leq k' < k$.  We aim at obtaining instances of $\apbc(n,k',k')$. To achieve this we reduce the size of the right side of the planted biclique to $k'$ in the same way as we reduced the size of the left side above: we will be replacing the vertices on the right side by randomly connected ones, one-by-one in a random order. As before we start with $M'_0 = M'$. To obtain $M'_{t+1}$, we choose a random and uniform column of $M'_t$ that was not previously picked and replace it with a random and uniform $\{0,1\}^n$ vector. This gives a sequence of random matrices: $M'_1, M'_2, \ldots, M'_n$. We know that for some $i^*$, $M'_{i^*}$ has a $k'\times k'$ biclique. We denote the right side of this biclique by $S^*$.  We now condition on the event that for all $i$, $M'_i$ does not contain a $k'\times k'$ biclique such that its left side is different from $L$. It is not hard to see that this event happens with probability at least $1-n 2^{-\Omega(k')} = 1-n 2^{-\Omega(k)}$.

Assume for now that we know $k'$. To recover $S$, we run $\A(n,k')$ on all the matrices $M'_i$. Let $L_i \times S_i$ be the biclique that $\A$ outputs on $M'_i$. We verify that $L_i \times S_i$ is a $k'\times k'$ biclique in $M'_i$. If so we let $S'$ be the set of all vertices on the right side connected to each vertex in $L_i$ (in the original graph after the permutation). If $|S'|=k$, we output $\pi^{-1}(S')$. Note that when executed on $M'_{i*}$, with probability $p'(n,k')$, this procedure will return $L \times S^*$.  Further, by our conditioning if the output of $\A$ is a $k'\times k'$ biclique then its left side must be $L$. All vertices in $\pi(S)$ are connected to $L$. The probability that any other vertex on the right side of $M'$ is connected to all vertices in $L$ is at most $n \cdot 2^{-k}$. Hence, conditioned on this event not occurring, we will recover exactly $S$.

To address the fact that $k'$ is not known, for each value of $k_1 = k-1,k-2,\ldots,k/2$, we run the algorithm under the assumption that $k' = k_1$ and stop once the algorithm has found a $k' \times k'$  biclique. If $k_1 > k'$ then, by our conditioning on none of $M'_i$ containing a $k'\times k'$ biclique such that its left side is different from $L$ there cannot exist a $k_1 \times k_1$ biclique in the graph. Therefore the algorithm will not output anything until $k_1 = k'$ at which point our analysis above applies.

To analyze the success probability and running time we can assume for simplicity that it is harder to find smaller planted bicliques than larger ones, and so for $k_1 \in  [k/2,k]$, $T'(n,k_1) \leq T'( n,k/2)$ and $p'(n,k')\leq p'(n,k/2)$. Therefore the running time of our algorithm is $T(n,k)=O(nk T'(n,k/2))$, and its success probability is $p(n,k)=p'(n,k/2) -  n2^{-\Omega(k)}$.
\end{proof}

We now prove the converse of Theorem \ref{thm:distributional2average}.
\begin{theorem}\label{thm:average2distributional}
Suppose that there is an algorithm that solves $\dpbc(n,k)$ that runs in time $T(n,k)$ and outputs the correct answer with probability $p(n,k)$. Then there exists an algorithm that
solves $\apbc(n,k',k')$  in time $T'(n,k')=O(nk' T(n,k'/2))$ and outputs the planted biclique with probability $p'(n,k') \geq p(n, k'/2)-n2^{\Omega(-k')}$.
\end{theorem}

\begin{proof}
Let $\A(n,k)$ denote the algorithm for solving $\dpbc(n,k)$, which, for any $S\subseteq [n]$ of size $k$, takes $n$ samples chosen according to  $D_S$ and outputs the planted set $S$ with probability $p(n,k)$. We will construct an algorithm for $\apbc(n,k',k')$ that takes as input an adjacency matrix $M$ chosen randomly according to $\cald_{avg}(n,k',k')$, (as in Definition \ref{apbc}) and outputs a biclique $S_1\times S_2$ of size  $k'\times k'$. Note that, with probability $1-n2^{-\Omega(k')}$, the set $S_1\times S_2$ is the unique $k'\times k'$ biclique in $M$ and we will condition on this event.

We first observe that an instance of $\dpbc(n,k)$ can be equivalently thought of as follows: first pick an arbitrary set $S$ of $k$ vertices on the right side of the graph; then pick $\ell$ according to $B(n,k/n)$; pick a random subset $S'$ of $\ell$ vertices on the left side of the graph; make $S' \times S$ a biclique and connect all the other pairs of vertices randomly and independently with probability $1/2$. The probability that $\ell \in [5k/6,6k/5]$ is at least $1-2^{-\Omega(k)}$ and therefore the probability that $\A(n,k)$ succeeds conditioned on this event is at least $p(n,k) - 2^{-\Omega(k)}$. This implies that there exists $\ell_k \in [5k/6,6k/5]$ such that conditioned on $\ell = \ell_k$, the probability that $\A(n,k)$ succeeds is at least $p(n,k) - 2^{-\Omega(k)}$ (we note that $\ell_k$ might depend on $S$).

Let $M$ be the adjacency matrix of the given instance of  $\apbc(n,k',k')$. For each column of $M$ (corresponding to a vertex on the right side), with probability $3/4$ we replace it with a random and uniform vector from $\zo^n$ and let $M'$ denote the obtained adjacency matrix. We denote by $S$ the subset of $S_2$ containing vertices that were not replaced by randomly connected vertices. Let $k=|S|$. With probability at least $1-2^{-\Omega(k)}$, $k \in [3k'/5, 5k'/6]$ and we will condition on this event.

We aim at obtaining an instance of $\dpbc(n,k)$ in which exactly $\ell_k$ vertices on the left are connected to all vertices in the planted set $S$. By the argument above, we know that we can assume that $\ell_k \in [5k/6,6k/5] \subseteq [k'/2,k']$ and $\A$ succeeds with probability at least $p(n,k) - 2^{-\Omega(k)}$ given an instance in which exactly $\ell_k$ vertices on the left are connected to all vertices in $S$.

 $M'$ has a $k' \times k$ biclique so to reduce the size of the left side of the planted biclique to $\ell_k$ we will be replacing the vertices on the left side by randomly connected ones, one-by-one in a random order as in the  proof of Theorem \ref{thm:distributional2average}.
This gives a sequence of random matrices: $M' = M'_0,M'_1, M'_2, \ldots, M'_n$. For every $i$, Let $S'_i$ denote the subset of vertices in $S_1$ that were not replaced by a randomly connected vertex in $M'_i$. It is also easy to see that for every $i$, conditioned on $|S'_i| = k''$, $M'_i$ is distributed exactly as $n$ samples from $D_S$ in which $k''$ samples were chosen to be connected to all vertices in $S$.

To recover $S_1$ and $S_2$, we run $\A(n,k)$ on all the matrices $M'_i$ (where, we assume for now that $k$ is known). Let $L_i \times R_i$ be the biclique that $\A$ outputs on $M'_i$. Let $S_1^*$ be the set of all (left side) vertices connected in the original input graph to all vertices in $R_i$ and $S_2^*$ be the set of all (right side) vertices in the input graph connected to all vertices in $S_1^*$. If $|S_1^*| = |S_2^*| = k'$ then we output the biclique $S_1^* \times S_2^*$. Otherwise we go to the next step (if none of the steps produces a biclique the algorithm fails).
We first note that, unless the algorithm fails, it outputs a $k'\times k'$ biclique in the input graph which, by our conditioning, can only be the true planted biclique. Further, there exists $i^*$ such that $|S'_{i^*}| =\ell_k$.  $M'_{i^*}$ is distributed as $n$ samples from $D_S$ in which $\ell_k$ samples were chosen to be connected to all vertices in $S$. Therefore by our conditioning, with probability at least $p(n,k) - 2^{-\Omega(k)}$, $\A(n,k)$ will output $S$. The vertices in $S_1$ are connected to all vertices in $S$ and with probability at least $1-n2^{-k}$ no other vertex in the original graph is. The vertices in $S_2$ are connected to all vertices in $S_1$ and, by our conditioning no other vertex is. Therefore, with probability at least  $p(n,k) - n2^{-\Omega(k)}$ the algorithm will produce the true $k' \times k'$ biclique.

To address the fact that $k$ is not known, for each value of $k_1 = \lceil 5k'/6 \rceil,\ldots, \lfloor 2k'/3\rfloor$, we run the algorithm under the assumption that $k = k_1$ and stop once the algorithm has found a $k' \times k'$  biclique. The algorithm can only output the true planted biclique and therefore this will not reduce the success probability.

As before, to analyze the success probability and running time we assume for simplicity that it is harder to find smaller planted sets than larger ones, and so for $k_1 \in [\lfloor 2k'/3\rfloor,  \lceil 5k'/6 \rceil]$, $T(n,k_1) \leq T(n,k'/2)$ and $p(n,k)\leq p(n,k'/2)$. Therefore the running time of our algorithm is $T'(n,k')=O(nk' T(n,k'/2))$, and its success probability is $p'(n,k')\geq p(n,k'/2) -  n2^{-\Omega(k')}$.
\end{proof}

\end{document}